\documentclass{amsproc}

\newtheorem{theorem}{Theorem}[section]
\newtheorem{lemma}[theorem]{Lemma}
\newtheorem{prop}[theorem]{Proposition}
\newtheorem{corollary}[theorem]{Corollary}

\theoremstyle{definition}

\theoremstyle{remark}
\newtheorem{remark}[theorem]{Remark}

\numberwithin{equation}{section}

\usepackage[sans]{dsfont}

 \renewcommand{\Re}{\mathrm{Re}}

  \renewcommand{\d}{\mathrm{d}}

\begin{document}

\title[On the Infrared Problem]{On the Infrared Problem for the Dressed Non-Relativistic Electron in a Magnetic Field}

\author[L. Amour]{Laurent Amour}
\address{Laboratoire de Math{\'e}matiques EDPPM, FRE-CNRS 3111, Universit{\'e} de Reims, Mou\-lin de la Housse - BP 1039, 51687 Reims Cedex 2, France}
\email{laurent.amour@univ-reims.fr}

\author[J. Faupin]{J{\'e}r{\'e}my Faupin}
\address{Institut de Math{\'e}matiques de Bordeaux, UMR-CNRS 5251, Universit{\'e} de Bordeaux 1, 351 cours de la Lib{\'e}ration, 33405 Talence Cedex, France}
\email{jeremy.faupin@math.u-bordeaux1.fr}

\author[B. Gr{\'e}bert]{Beno{\^i}t Gr{\'e}bert}
\address{Laboratoire de Math{\'e}matiques Jean Leray, UMR-CNRS 6629, Universit{\'e} de Nantes, 2 rue de la Houssini{\`e}re, 44072 Nantes Cedex 3, France}
\email{benoit.grebert@univ-nantes.fr}

\author[J.-C. Guillot]{Jean-Claude Guillot}
\address{Centre de Math{\'e}matiques appliqu{\'e}es, UMR-CNRS 7641, Ecole polytechnique, 99128 Palaiseau Cedex, France}
\email{jean-claude.guillot@polytechnique.edu}

\date{}

\begin{abstract}
We consider a non-relativistic electron interacting with a classical magnetic field pointing along the $x_3$-axis and with a quantized electromagnetic field. The system is translation invariant in the $x_3$-direction and we consider the reduced Hamiltonian $H(P_3)$ associated with the total momentum $P_3$ along the $x_3$-axis. For a  fixed momentum $P_3$ sufficiently small, we prove that $H(P_3)$ has a ground state in the Fock representation if and only if $E'(P_3)=0$, where $P_3 \mapsto E'(P_3)$ is the derivative of the map $P_3 \mapsto E(P_3) = \inf \sigma (H(P_3))$. If $E'(P_3) \neq 0$, we obtain the existence of a ground state in a non-Fock representation. This result holds for sufficiently small values of the coupling constant. \textbf{MSC}: 81V10; 81Q10; 81Q15
\end{abstract}

\maketitle

\section{Introduction}

In this paper we pursue the analysis of a model considered in \cite{AGG1}, describing a non-relativistic particle (an electron) interacting both with the quantized electromagnetic field and a classical magnetic field pointing along the $x_3$-axis. An ultraviolet cutoff is imposed in order to suppress the interaction between the electron and the photons of energies bigger than a fixed, arbitrary large parameter $\Lambda$. The total system being invariant by translations in the $x_3$-direction, it can be seen (see \cite{AGG1}) that the corresponding Hamiltonian admits a decomposition of the form $H \simeq \int^\oplus_\mathbb{R} H(P_3) dP_3$ with respect to the spectrum of the total momentum along the $x_3$-axis that we denote by $P_3$. For any given $P_3$ sufficiently close to 0, the existence of a ground state for $H(P_3)$ is proven in \cite{AGG1} provided an infrared regularization is introduced (besides a smallness assumption on the coupling parameter). Our aim is to address the question of the existence of a ground state without requiring any infrared regularization.

The model considered here is closely related to similar non-relativistic QED models of freely moving electrons, atoms or ions, that have been studied recently (see \cite{BCFS,FGS1,Hir,CF,Chen,HH,CFP,FP} for the case of one single electron, and \cite{AGG2,LMS,FGS2,HH,LMS2} for atoms or ions). In each of these papers, the physical systems are translation invariant, in the sense that the associated Hamiltonian $H$ commutes with the operator of total momentum $P$. As a consequence, $H \simeq \int_{\mathbb{R}^3} H(P) dP$, and one is led to study the spectrum of the fiber Hamiltonian $H(P)$ for fixed $P$'s.

For the one-electron case, an aspect of the so-called \emph{infrared catastrophe} lies in the fact that, for $P \neq 0$, $H(P)$ does not have a ground state in the Fock space (see \cite{CF,Chen,HH,CFP}). More precisely, if an infrared cutoff of parameter $\sigma$ is introduced in the model in order to remove the interaction between the electron and the photons of energies less than $\sigma$, the associated Hamiltonian $H_\sigma(P)$ does have a ground state $\Phi_\sigma(P)$ in the Fock space. Nevertheless as $\sigma \rightarrow 0$, it is shown that $\Phi_\sigma(P)$ ``leaves'' the Fock space. Physically this can be interpreted by saying that a free moving electron in its ground state is surrounded by a cloud of infinitely many ``soft'' photons.

For negative ions, the absence of a ground state for $H(P)$ is established in \cite{HH} under the assumption $\nabla E(P) \neq 0$, where $E(P) = \inf \sigma(H(P))$.

In \cite{CF}, with the help of operator-algebra methods, a representation of a \emph{dressed 1-electron state} non-unitarily equivalent to the usual Fock representation of the canonical commutation relations is given. We shall obtain in this paper a related result, following a different approach, under the further assumption that the electron interact with a classical magnetic field and an electrostatic potential.

We shall first provide a necessary and sufficient condition for the existence of a ground state for $H(P_3)$.  Namely we shall prove that the bottom of the spectrum, $E(P_3) = \inf \sigma(H(P_3))$, is an eigenvalue of $H(P_3)$ if and only if $E'(P_3) = 0$ where $E'(P_3)$ denotes the derivative of the map $P_3 \mapsto E(P_3)$. In the case $E'(P_3) \neq 0$, thanks to a (non-unitary) Bogoliubov transformation, in the same way  as in \cite{Arai,DG2}, we shall define a ``renormalized'' Hamiltonian $H^{\mathrm{ren}}(P_3)$ which can be seen as an expression of the physical Hamiltonian in a non-Fock representation. Then we shall prove that $H^{\mathrm{ren}}(P_3)$ has a ground state. These results have been announced in \cite{AFGG}.

The regularity of the map $P_3 \mapsto E(P_3)$ plays a crucial role in our proof. Adapting \cite{Pizzo,CFP} we shall see that $P_3 \mapsto E(P_3)$ is of class $\mathrm{C}^{1+\gamma}$ for some strictly positive $\gamma$. Let us also mention that our method can be adapted to the case of free moving hydrogenoid ions without spin, the condition $E'(P_3)=0$ being replaced by $\nabla E(P) = 0$ (see Subsection \ref{subsection:results} for a further discussion on this point).

The remainder of the introduction is organized as follows: In Subsection \ref{subsection:model}, a precise definition of the model considered in this paper is given, next, in Subsection \ref{subsection:results}, we state our results and compare them to the literature.

\subsection{The model}\label{subsection:model}

We consider a non-relativistic electron of charge $e$ and mass $m$ interacting with a classical magnetic field pointing along the $x_3$-axis, an electrostatic potential, and the quantized electromagnetic field in the Coulomb gauge. The Hilbert space for the electron and the photon field is written as
\begin{equation}
\mathcal{H} = \mathcal{H}_{\rm el} \otimes \mathcal{H}_{\rm ph},
\end{equation}
where $\mathcal{H}_{\rm el} = \mathrm{L}^2( \mathbb{R}^3 ; \mathbb{C}^2 )$ is the Hilbert space for the electron, and $\mathcal{H}_{\rm ph}$ is the symmetric Fock space over $\mathrm{L}^2( \mathbb{R}^3 \times \mathbb{Z}_2 )$ for the photons,
\begin{equation}
\mathcal{H}_{\mathrm{ph}} = \mathbb{C} \oplus \bigoplus_{n=1}^\infty S_n \left [ \mathrm{L}^2( \mathbb{R}^3 \times \mathbb{Z}_2 )^{ \otimes^n} \right ].
\end{equation}
Here $S_n$ denotes the orthogonal projection onto the subspace of symmetric functions in $\mathrm{L}^2( \mathbb{R}^3 \times \mathbb{Z}_2 )^{ \otimes^n}$ in accordance with Bose-Einstein statistics. We shall use the notation $\bold k = (k,\lambda)$ for any $(k,\lambda) \in \mathbb{R}^3 \times \mathbb{Z}_2$, and
\begin{equation}
\int_{\mathbb{R}^3 \times \mathbb{Z}_2 }d\bold k = \sum_{\lambda=1,2} \int_{\mathbb{R}^3} dk.
\end{equation}
Likewise, the scalar product in $\mathrm{L}^2( \mathbb{R}^3 \times \mathbb{Z}_2 )$ is defined by
\begin{equation}
(h_1,h_2) = \int_{ \mathbb{R}^3 \times \mathbb{Z}_2 } \bar h_1(\bold k) h_2(\bold k) d\bold k = \sum_{\lambda=1,2} \int_{\mathbb{R}^3} \bar h_1(k,\lambda) h_2(k,\lambda) dk.
\end{equation}
The position and the momentum of the electron are denoted respectively by $x=(x_1,x_2,x_3)$ and $p=(p_1,p_2,p_3)=-i\nabla_x$. The classical magnetic field is of the form $(0,0,b(x'))$, where $x'=(x_1,x_2)$ and $b(x') = ( \partial a_2/ \partial x_1 ) (x') - ( \partial a_1 / \partial x_2 )(x')$. Here $a_j(x')$, $j=1,2$, are real functions in $\mathrm{C}^1( \mathbb{R}^2)$. The electrostatic potential is denoted by $V(x')$. The quantized electromagnetic field in the Coulomb gauge is defined by 
\begin{equation}\label{eq:def_A_B}
\begin{split}
& A(x)= \frac{1}{\sqrt{2\pi}} \int \frac{ \epsilon^\lambda(k) }{ |k|^{1/2} } \rho^\Lambda(k) \Big [
 e^{-ik\cdot x} a^*(\bold k) + e^{ik\cdot x}a(\bold k)\Big ] d \bold k, \\
& B(x)= -\frac{i}{\sqrt{2\pi}} \int |k|^{1/2} \left(\frac{k}{|k|}\wedge\epsilon^{\lambda}(k)\right) \rho^\Lambda(k) \Big [
 e^{-ik\cdot x} a^*(\bold k) - e^{ik\cdot x}a(\bold k)\Big ] d\bold k,
\end{split}
\end{equation}
where $\rho^\Lambda(k)$ denotes the characteristic function $\rho^\Lambda(k) = \mathds{1}_{|k| \le \Lambda}(k)$ and $\Lambda$ is an arbitrary large positive real number. Note that this explicit choice of the ultraviolet cutoff function $\rho^\Lambda$ is made mostly for convenience. Our results would hold without  change for any $\rho^\Lambda$ satisfying $\int_{|k|\le 1} |k|^{-2} |\rho^\Lambda(k)|^2 d^3k + \int_{|k|\ge 1} |k| |\rho^\Lambda(k)|^2d^3k < \infty$. The vectors $\epsilon^1(k)$ and $\epsilon^2(k)$ in \eqref{eq:def_A_B} are real polarization vectors orthogonal to each other and to $k$. Besides $a^*(\bold k)$ and $a(\bold k)$ are the usual creation and annihilation operators obeying the canonical commutation relations
\begin{equation}\label{eq:CCR}
\left [ a^\#(\bold k) , a^\#(\bold k') \right ] = 0 \quad , \quad \left [ a(\bold k) , a^*(\bold k') \right ] = \delta( \bold k - \bold k' ) = \delta_{\lambda\lambda'} \delta( k - k' ).
\end{equation}
The Pauli Hamiltonian $H_g$ associated with the system we consider is formally given by 
\begin{equation}\label{eq:defH}
\begin{split}
H_g =& \frac{1}{2m} \Big ( p - e a(x') - g A(x) \Big )^2 - \frac{e}{2m} \sigma_3 b(x') \\
& - \frac{g}{2m} \sigma \cdot B(x) + V(x') + H_{\rm ph},
\end{split}
\end{equation}
where the charge of the electron is replaced by a coupling parameter $g$ in the terms containing the quantized electromagnetic field. The Hamiltonian for the photons in the Coulomb gauge is given by 
\begin{equation}\label{eq:def_Hph}
H_{\rm ph} = \d \Gamma( |k| ) = \int |k| a^*(\bold k) a(\bold k) d \bold k.
\end{equation}
Finally $\sigma=(\sigma_1,\sigma_2,\sigma_3)$ is the 3-component vector of the Pauli matrices.

Noting that $H_g$ formally commutes with the operator of total momentum in the direction $x_3$, $P_3 = p_3 + {\rm d} \Gamma (k_3)$, one can consider the reduced Hamiltonian associated with $P_3 \in \mathbb{R}$ that we denote by $H_g(P_3$). For any fixed $P_3$, $H_g(P_3)$ acts on $\mathrm{L}^2( \mathbb{R}^2 ; \mathbb{C}^2 ) \otimes \mathcal{H}_{\rm ph}$ and is formally given by
\begin{equation}\label{eq:H(P3)}
\begin{split}
H_g (P_3)=&
     \frac{1}{2m} \sum_{j=1,2} \Big ( p_j - e a_j(x') - g A_j(x',0) \Big )^2 - \frac{e}{2m} \sigma_3 b(x') + V(x') \\
&+ \frac{1}{2m} \Big ( P_3 - {\rm d} \Gamma(k_3) - g A_3(x',0) \Big )^2 - \frac{g}{2m} \sigma \cdot B(x',0) + H_{\rm ph}. 
\end{split}
\end{equation}

We define the infrared cutoff Hamiltonian $H_{g}^\sigma(P_3)$ by replacing $A(x)$ in \eqref{eq:def_A_B} with
\begin{equation}\label{eq:def_A_Bsigma}
\begin{split}
& A_\sigma(x)= \frac{1}{\sqrt{2\pi}} \int \frac{ \epsilon^\lambda(k) }{ |k|^{1/2} } \rho_\sigma^\Lambda(k) \Big [ e^{-ik\cdot x} a^*(\bold k) + e^{ik\cdot x}a(\bold k)\Big ] d \bold k, 
\end{split}
\end{equation}
where $\rho^\Lambda_\sigma = \mathds{1}_{\sigma \le |k| \le \Lambda}$, and similarly for $B_\sigma(x)$. We set $E_g(P_3) = \inf \sigma( H_g(P_3) )$ and $E_{g\sigma}(P_3) = \inf \sigma( H_{g}^\sigma(P_3) )$.
 
The electronic Hamiltonian $h(b,V)$ on $\mathrm{L}^2(\mathbb{R}^2 ; \mathbb{C}^2 )$ is defined by 
\begin{equation}
h(b,V) = \sum_{j=1,2} \frac{1}{2m} ( p_j - e a_j(x') )^2 - \frac{e}{2m} \sigma_3 b(x') + V(x').
\end{equation}
Let $e_0 = \inf \sigma( h(b,V) )$.
We make the following hypothesis:
\begin{itemize}
\item[{\bf ($\mathcal{\mathbf{H}}_\mathbf{0}$)}] $h(b,V)$ is essentially self-adjoint on $\mathrm{C}_0^\infty( \mathbb{R}^2 ; \mathbb{C}^2 )$ and $e_0$ is an isolated eigenvalue of multiplicity 1.
\end{itemize}
We refer to \cite{AHS,Sobolev,IT,Raikov} for possible choices of $b,V$ satisfying Hypothesis {\bf ($\mathcal{\mathbf{H}}_\mathbf{0}$)}. The following proposition is established in \cite[Theorem 2.3]{AGG1}:
\begin{prop}\label{prop:self-adjointness}
Suppose Hypothesis $\mathbf{(H_0)}$. For sufficiently small values of $|g|$, $H_g$ is self-adjoint with domain $D(H_g)=D(H_0)$, and for any $\sigma \ge 0$ and $P_3 \in \mathbb{R}$,  $H_g^\sigma(P_3)$ identifies with a self-adjoint operator with domain $D(H_g^\sigma(P_3)) = D(H_0(P_3))$. Moreover $H_g$ admits the decomposition
\begin{equation}
H_g = \int^\oplus_{\mathbb{R}} H_g(P_3) dP_3.
\end{equation}
\end{prop}

\subsection{Results and comments}\label{subsection:results}

The key ingredient that we shall need in order to prove our main theorem (see Theorem \ref{thm:main} below) lies in the regularity of the map $P_3 \mapsto E'_{g\sigma}(P_3)$ uniformly in $\sigma \ge 0$.
\begin{theorem}\label{thm:regularity}
Assume that $\mathbf{(H_0)}$ holds. There exists $g_0>0$, $\sigma_0>0$ and $P_0>0$ such that for all $|g| \le g_0$, for all $0 \le \sigma \le \sigma_0$, for all $P_3,k_3$ such that $|P_3| \le P_0$, $| P_3+k_3| \le P_0$, for all $\delta>0$,
\begin{equation}\label{eq:regularity}
| E'_{g\sigma}(P_3+k_3) - E'_{g\sigma}(P_3) | \le \mathrm{C}_\delta |k_3|^{\frac{1}{4}-\delta},
\end{equation}
where $\mathrm{C}_\delta$ is a positive constant depending only on $\delta$.
\end{theorem}
Similar results for a free electron (that is for $b=V=0$) interacting with the quantized electromagnetic field have been obtained recently (see \cite{Chen,CFP,FP}). The model studied in the latter papers is technically simpler than the one considered here in that the fiber Hamiltonian $H(P)$ associated with a free electron does not contain the electronic part $h(b,V)$ and its (minimal) coupling to the quantized electromagnetic field. In particular the operator $H(P)$ in \cite{Chen, CFP, FP} acts only on the Fock space, whereas in our case $H_{g\sigma}(P_3)$ still contains interactions between the electromagnetic field and the electronic degrees of freedom. We shall use the exponential decay of the ground states $\Phi_g^\sigma(P_3)$ in $x'$ in order to overcome this difficulty.

It is proved in \cite{Chen} (for a free electron) that $P \mapsto E(P) = \inf \sigma( H(P) )$ is of class $\mathrm{C}^2$ in a neighborhood of 0 thanks to a renormalization group analysis (see also \cite{BCFS}). The author also shows that, still in a neighborhood of $P=0$, the derivative $\nabla E(P)$ vanishes only at $P=0$. In \cite{CFP}, with the help of what the authors call ``iterative analytic perturbation theory'', following a multiscale analysis developed in \cite{Pizzo}, it is proved, among other results, that $P \mapsto E(P)$ is of class $\mathrm{C}^{5/4-\delta}$ for arbitrary small $\delta>0$. The method has later been improved in \cite{FP} leading to the $\mathrm{C}^2$ property of $P \mapsto E(P)$.

In order to establish our main theorem, Theorem \ref{thm:main}, the ``degree of regularity'' we need is reached as soon as $P_3 \mapsto E_{g\sigma}(P_3)$ is at least of order $\mathrm{C}^{1+\gamma}$, uniformly in $\sigma$, for some $\gamma>0$. Theorefore, although one can conjecture that $P_3 \mapsto E_{g\sigma}(P_3)$ is of class $\mathrm{C}^2$ uniformly in $\sigma$, Theorem \ref{thm:regularity} is sufficient for our purpose. In order to prove it we shall adapt \cite{Pizzo,CFP}: First, we shall give a short proof of the existence of a spectral gap for $H_{g}^\sigma(P_3)$ (restricted to the space of photons of energies bigger than $\sigma$) above the non-degenerate eigenvalue $E_{g\sigma}(P_3)$. Next we shall apply ``iterative analytic perturbation theory''.

We postpone the proof of Theorem \ref{thm:regularity} to the appendix. Since several parts are taken from \cite{Pizzo,CFP}, we shall not give all the details, rather we shall emphasize the differences with \cite{Pizzo,CFP}.

For $h \in \mathrm{L}^2( \mathbb{R}^3 \times \mathbb{Z}_2)$, let us define the field operator $\Phi(h)$ by
\begin{equation}
\Phi(h) = \frac{1}{\sqrt{2}} (a^*(h) + a(h)),
\end{equation}
where the creation operator $a^*(h)$ and the annihilation operator $a(h)$ are defined respectively by
\begin{equation}
a^*(h) = \int_{\mathbb{R}^3 \times \mathbb{Z}_2 } h( \bold k ) a^*( \bold k ) d \bold k, \quad a(h) = \int_{\mathbb{R}^3 \times \mathbb{Z}_2 } \bar h( \bold k ) a( \bold k ) d \bold k.
\end{equation}
Hence, letting $h_{j,\sigma}(x')$ and $\tilde h_{j,\sigma}(x')$ for $j=1,2,3$ be defined respectively by
\begin{equation}\label{eq:h_j_tildeh_j}
\begin{split}
& h_{j,\sigma}(x',\bold k) = \pi^{-1/2} \frac{ \epsilon^\lambda_j(k) }{ |k|^{1/2} } \rho^\Lambda_\sigma(k)e^{ik' \cdot x'}, \\
& \tilde h_{j,\sigma}(x',\bold k) = -i\pi^{-1/2} |k|^{1/2} \left ( \frac{k}{|k|} \wedge \epsilon^\lambda(k) \right )_j \rho^\Lambda_\sigma(k) e^{ik' \cdot x'},
\end{split}
\end{equation}
where $k'=(k_1,k_2)$, we have $A_{j,\sigma}(x',0) = \Phi( h_{j,\sigma}(x') )$ and $B_{j,\sigma}(x',0) = \Phi( \tilde h_{j,\sigma}(x') )$.
The Weyl operator associated with $h \in \mathrm{L}^2( \mathbb{R}^3 \times \mathbb{Z}_2 )$ is denoted by $W(h) = e^{i \Phi(h)}$.
Let $f_\sigma : \mathbb{R}^3 \times \mathbb{Z}_2 \rightarrow \mathbb{C}$ be defined by
\begin{equation}\label{eq:def_fsigma}
f_\sigma( \bold k ) = - \frac{g}{\sqrt{4\pi}} \frac{ \rho^\Lambda_\sigma(k) \epsilon_\lambda^3(k) }{ k_3|k|^{1/2} } \frac{ E_{g\sigma}(P_3-k_3) - E_{g\sigma}(P_3) }{ E_{g\sigma}(P_3-k_3) - E_{g\sigma}(P_3) + |k| }.
\end{equation}
If $\sigma=0$ we remove the subindex $\sigma$ in the preceding notations. We recall from \cite[Lemma 4.3]{AGG1} that for $g$, $\sigma$, $P_3$ and $|k|$ sufficiently small,
\begin{equation}\label{eq:tildeE-E}
E_{g\sigma}(P_3-k_3) - E_{g\sigma}(P_3) \ge - \frac{3}{4} |k|.
\end{equation}
Hence in particular for $\sigma > 0$, we have $f_\sigma \in \mathrm{L}^2( \mathbb{R}^3 \times \mathbb{Z}_2$), whereas if $\sigma=0$ and $P_3 \mapsto E_g(P_3)$ is of class $\mathrm{C}^{1+\gamma}$ with $\gamma>0$, then
\begin{equation}\label{eq:f_in_L2}
f \in \mathrm{L}^2( \mathbb{R}^3 \times \mathbb{Z}_2 ) \iff E'_g(P_3) = 0.
\end{equation}
Similarly as in \cite{Arai} (see also \cite{DG2,Panati}), we define the ``renormalized'' (Bogoliubov transformed) Hamiltonian $H_{g\sigma}^{\mathrm{ren}}(P_3)$ by the expression
\begin{equation}\label{eq:formaldef_Hren}
H_{g\sigma}^{\mathrm{ren}}(P_3) = W(if_\sigma) H_{g}^\sigma(P_3) W(if_\sigma)^*.
\end{equation}
Notice that the identity \eqref{eq:formaldef_Hren} might only be formal for $\sigma=0$ since in this case, by \eqref{eq:f_in_L2}, $f$ might not be in $\mathrm{L}^2$. Nevertheless using usual commutation relations (see for instance \cite{DG1}), we define for any $\sigma \ge 0$:
\begin{equation*}
\begin{split}
&H^{\rm ren}_{g\sigma}(P_3) = \frac{1}{2m} \sum_{j=1,2} \Big ( p_j - e a_j( x' ) - g A_{j,\sigma}(x',0) + g \Re ( h_{j,\sigma}(x') , f_\sigma )  \Big )^2 \\
& + \frac{1}{2m} \Big ( P_3 - \mathrm{d} \Gamma(k_3) - \Phi(k_3f_\sigma) - \frac{1}{2} ( k_3f_\sigma,f_\sigma ) - gA_{3,\sigma}(x',0) + g \Re ( h_{3,\sigma}(x') , f_\sigma ) \Big )^2 \\
& - \frac{e}{2m} \sigma_3b(x') - \frac{g}{2m} \sigma \cdot \Big ( B_\sigma(x',0) - \Re ( \tilde h_{\sigma}(x') , f_{\sigma} ) \Big ) + V(x') \\
& + H_f + \Phi ( |k| f_\sigma ) + \frac{1}{2} ( |k| f_\sigma , f_\sigma ).
\end{split}
\end{equation*}
In the same way as for $H_g^\sigma(P_3)$ (see Proposition \ref{prop:self-adjointness}), one can verify that $H^{\mathrm{ren}}_{g\sigma}(P_3)$ is self-adjoint with domain $D( H^{\mathrm{ren}}_{g\sigma}(P_3) ) = D( H_0(P_3) )$ for any $\sigma \ge 0$. Besides for $\sigma > 0$, we have that $H_{g\sigma}^{\mathrm{ren}}(P_3)$ is unitarily equivalent to $H_{g}^\sigma(P_3)$, whereas for $\sigma=0$, one can verify that $H_g^{\mathrm{ren}}(P_3)$ is unitarily equivalent to $H_g(P_3)$ if and only if $f \in \mathrm{L}^2( \mathbb{R}^3 \times \mathbb{Z}_2 )$. Our main result is:
\begin{theorem}\label{thm:main}
Suppose Hypothesis $\mathbf{(H_0)}$. There exist $g_0>0$ and $P_0>0$ such that for all $0 \le |g| \le g_0$ and $0 \le |P_3| \le P_0$,
\begin{itemize}
\item[(i)] $H_g(P_3)$ has a ground state if and only if $E'_g(P_3) = 0$,
\item[(ii)] $H^{\mathrm{ren}}_g(P_3)$ has a ground state.
\end{itemize}
\end{theorem}
The proof of Theorem \ref{thm:main} can be adapted to the case of free moving hydrogenoid ions without spins\footnote{The hypothesis of simplicity for the electronic ground state $\mathbf{(H_0)}$ imposes this restriction to hydrogenoid atoms or ions.}, the condition $E'_g(P_3)=0$ being replaced by $\nabla E_g(P) = 0$, where $E_g(P)$ denotes the bottom of the spectrum of the fiber Hamiltonian $H_g(P)$. The existence of ground states for atoms has been obtained in \cite{AGG2} thanks to a Power-Zienau-Wooley transformation and the crucial property $Q=0$ (here $Q$ denotes the total charge of the atomic system). Indeed, in \cite{HH}, it is proved that for negative ions ($Q<0$) $H_g(P)$ does not have a ground state if $\nabla E_g(P) \neq 0$. Let us also mention \cite{LMS} where the existence of ground states for atoms is proven for any value of the coupling constant $g$, by adapting \cite{GLL}, under the further assumption $E_g(P) \ge E_g(0)$ which has not been proven yet. Thus in addition to these results, our method provides the existence of ground states for spinless hydrogenoid ions, both for $H_g(P)$ in the case $\nabla E_g(P)=0$ and for $H^{\mathrm{ren}}_g(P)$.

The two statements ``$H_g(P_3)$ has a ground state if $E'_g(P_3)=0$'' and ``$H^{\mathrm{ren}}_g(P_3)$ has a ground state'' shall be established following the same standard procedure: An infrared cutoff $\sigma$ is introduced into the model so that the Hamiltonian $H_{g}^\sigma(P_3)$ (respectively $H_{g\sigma}^{\mathrm{ren}}(P_3)$) has a ground state $\Phi_g^\sigma(P_3)$ (respectively $\Phi_{g\sigma}^{\mathrm{ren}}(P_3)$). We then need to prove that $\Phi_g^\sigma(P_3)$ and $\Phi_{g\sigma}^{\mathrm{ren}}(P_3)$ converge strongly as $\sigma \rightarrow 0$. To this end we control the number of photons in the states $\Phi_g^\sigma(P_3)$ and $\Phi_{g\sigma}^{\mathrm{ren}}(P_3)$ thanks to a pull-through formula and \eqref{eq:regularity}.

We emphasize that, in the case $E'_g(P_3) \neq 0$, $H^{\mathrm{ren}}_g(P_3)$ can be seen as an expression of the physical Hamiltonian in a representation of the canonical commutation relations non-unitarily equivalent to the Fock representation. Besides, regarding \cite{Chen} for the case of a single freely moving electron, one can conjecture that for sufficiently small values of $|P_3|$, $E'_g(P_3) = 0$ if and only if $P_3 = 0$.

Our proof of the absence of a ground state for $H_g(P_3)$ in the case $E'_g(P_3) \neq 0$ is based on a contradiction argument and \cite[Lemma 2.6]{DG2} (see also Lemma \ref{lm:DG}). Again the result is achieved by deriving a suitable expression of $a( \bold k ) \Phi_g(P_3)$ thanks to a pull-through formula (assuming here that $H_g(P_3)$ has a ground state $\Phi_g(P_3)$). Note that the regularity property \eqref{eq:regularity} appears again as a key property (although here only \eqref{eq:regularity} for $\sigma = 0$ is required).

The paper is organized as follows: In Section \ref{section:main}, we prove Theorem \ref{thm:main}. Next in the appendix we prove Theorem \ref{thm:regularity}.

\noindent \textbf{Acknowledgements}. We would like to thank G. Raikov for useful comments.

\section{Proof of Theorem \ref{thm:main}}\label{section:main}

The following proposition is proven in Subsection \ref{subsection:gap} of the appendix.
\begin{prop}\label{prop:GS}
Assume that $\mathbf{(H_0)}$ holds. There exists $g_0>0$, $\sigma_0>0$ and $P_0>0$ such that for all $|g|\le g_0$, for all $0<\sigma\le\sigma_0$, for all $|P_3| \le P_0$, $H_{g\sigma}(P_3)$ has a unique normalized ground state $\Phi_{g}^\sigma(P_3)$, i.e.
\begin{equation}
H_{g}^\sigma(P_3) \Phi_{g}^\sigma(P_3) = E_{g\sigma}(P_3) \Phi_{g}^\sigma(P_3), \quad \| \Phi_{g}^\sigma(P_3) \| = 1.
\end{equation}
\end{prop}
Notice that Proposition \ref{prop:GS} is also established in \cite{AGG1} under the weaker assumption that $e_0$ is an isolated eigenvalue of $h(b,V)$ of finite multiplicity. Let us recall a lemma, due to \cite{DG2}, on which is based our proof of the absence of a ground state for $H_g(P_3)$ in the case $E'_g(P_3) \neq 0$.
\begin{lemma}\label{lm:DG}
Let $\Psi \in \mathrm{L}^2( \mathbb{R}^2 ; \mathbb{C}^2 ) \otimes \mathcal{H}_{\rm ph}$. Assume that
\begin{equation}
\int_{\mathbb{R}^3 \times \mathbb{Z}_2} \| ( a( \bold k ) - h( \bold k ) ) \Psi \|^2 d \bold k < \infty,
\end{equation}
where $h$ is a measurable function from $\mathbb{R}^3 \times \mathbb{Z}_2$ to $\mathbb{C}$ such that $h \notin \mathrm{L}^2( \mathbb{R}^3 \times \mathbb{Z}_2 )$. Then $\Psi = 0$.
\end{lemma}
\begin{proof}
See \cite[Lemma 2.6]{DG2}.
\end{proof}
Theorem \ref{thm:main} shall follow from a suitable decomposition of $a( \bold k ) \Phi_{g}^\sigma(P_3)$ based on a pull-through formula. The latter is the purpose of the following lemma, where the equalities should be understood as identities between measurable functions from $\mathbb{R}^3 \times \mathbb{Z}_2$ to $\mathrm{L}^2( \mathbb{R}^2 ; \mathbb{C}^2 ) \otimes \mathcal{H}_{\rm ph}$. For a rigorous justification of the commutations used in the next proof, we refer for instance to \cite{Ge,HH}.

In order to shorten the notations, we shall write
\begin{equation}\label{eq:notations}
\begin{split}
& H = H_{g}^\sigma(P_3), \quad E = E_{g\sigma}(P_3), \quad \Phi = \Phi_{g}^\sigma(P_3), \\
&\tilde H = H_{g}^\sigma(P_3-k_3), \quad \tilde E = E_{g\sigma}(P_3-k_3).
\end{split}
\end{equation}
\begin{lemma}\label{lm:pull-through}
Let $\sigma \ge 0$ and let $\Phi = \Phi_{g}^\sigma(P_3)$ be a normalized ground state of $H = H_{g}^\sigma(P_3)$ (assuming it exists for $\sigma=0$). We have:
\begin{equation}
a(\bold k) \Phi = L_\sigma(\bold k) \Phi + R_\sigma(\bold k) \Phi + \frac{1}{\sqrt{2}} f_\sigma( \bold k ) \Phi,
\end{equation}
where $L_\sigma$ and $R_\sigma$ are operator-valued functions such that,
\begin{equation}\label{eq:estimate_Lsigma}
\int_{\mathbb{R}^3 \times \mathbb{Z}_2 } \| L_\sigma( \bold k ) \Phi \|^2 d \bold k \le \mathrm{C} g^2,
\end{equation}
and
\begin{equation}\label{eq:def_R}
R_\sigma( \bold k ) = - \frac{g}{2\sqrt{2\pi}} \frac{ \rho^\Lambda_\sigma( k ) \epsilon^\lambda_3(k) |k|^{1/2} }{ k_3 ( \tilde E - E + |k| ) } \frac{ \tilde H - \tilde E }{ \tilde H - E + |k| }.
\end{equation}
\end{lemma}
\begin{proof}
It follows from the canonical commutation relations \eqref{eq:CCR} that
\begin{equation}\label{eq:pull-through1}
\begin{split}
&a( \bold k ) H = ( \tilde H + |k| ) a( \bold k ) \\
& - \frac{g}{2^{\frac{3}{2}}m} \sum_{j=1,2} \Big ( h_{j,\sigma}(x',\bold k) \big ( p_j - e a_j(x') - g A_{j,\sigma}(x',0) \big ) + \sigma_j \tilde h_{j,\sigma}(x',\bold k) \Big )\\
& - \frac{g}{2^{\frac{3}{2}}m} \Big ( h_{3,\sigma}(x',\bold k) \big ( P_3 - \d \Gamma( k_3 ) - g A_{3,\sigma}(x',0) \big ) + \sigma_3 \tilde h_{3,\sigma}(x',\bold k) \Big ).
\end{split}
\end{equation}
In order to control the term containing $( p_j - e a_j(x') - g A_{j,\sigma}(x',0))$ in the right-hand-side of the previous equality, we use that (formally)
\begin{equation}\label{eq:commutator_xj}
\frac{1}{2m} \big ( p_j - e a_j(x') - g A_{j,\sigma}(x',0) \big ) = i [ H , x'_j ],
\end{equation}
for $j=1,2$. Notice that an alternative would be to consider the Hamiltonian obtained through a unitary Power-Zienau-Wooley transformation (see for instance \cite{GLL}). For a rigorous justification of \eqref{eq:commutator_xj}, we refer to \cite[Theorem II.10]{BFP} which can easily be adapted to our case. In particular it follows that $x'_j \Phi \in D(H)$. Applying \eqref{eq:pull-through1} to $\Phi$ then yields
\begin{equation}\label{eq:pull-through2}
\begin{split}
& a( \bold k ) \Phi = \frac{i g}{2^{\frac{1}{2}}} \sum_{j=1,2} h_{j,\sigma}(x',\bold k) [\tilde H - E + |k| ]^{-1} ( H - E ) x'_j \Phi \\
&  + \frac{g}{2^{\frac{3}{2}}m} [ \tilde H - E + |k| ]^{-1} \sigma \cdot \tilde h _\sigma(x',\bold k) \Phi\\
& + \frac{g}{2^{\frac{3}{2}}m} h_{3,\sigma}(x',\bold k) [ \tilde H - E + |k| ]^{-1} \big ( P_3 - \d \Gamma( k_3 ) - g A_{3,\sigma}(x',0) \big ) \Phi.
\end{split}
\end{equation}
Note that the expressions of $H$ and $\tilde H$ imply
\begin{equation}\label{eq:tildeH-H}
\tilde H - H = - \frac{k_3}{m} \big ( P_3 - \d \Gamma( k_3 ) - g A_{3,\sigma}(x',0) \big ) + \frac{k_3^2}{2m}.
\end{equation}
From \eqref{eq:tildeE-E}, we get
\begin{equation}\label{eq:pullthrough3}
\| [ \tilde H - E + |k| ]^{-1} \| \le \mathrm{C} |k|^{-1}.
\end{equation}
Moreover it is not difficult to show that
\begin{equation}\label{eq:pullthrough4}
\left \| \big ( P_3 - \d \Gamma( k_3 ) - g A_{3,\sigma}(x',0) \big ) [ \tilde H - E + |k| ]^{-1} \right \| \le \mathrm{C} |k|^{-1},
\end{equation}
and consequently, by \eqref{eq:tildeH-H},
\begin{equation}\label{eq:pullthrough5}
\left \| (H-E) [\tilde H - E + |k| ]^{-1} \right \| \le \mathrm{C}.
\end{equation}
Introducing \eqref{eq:pullthrough3}--\eqref{eq:pullthrough5} into \eqref{eq:pull-through2} and recalling the definitions \eqref{eq:h_j_tildeh_j} of $h_j$ and $\tilde h_j$, we thus obtain
\begin{equation}\label{eq:pullthrough6}
\begin{split}
&a( \bold k ) \Phi = L_1( \bold k )\Phi \\
&+ \frac{g}{2^{\frac{3}{2}}m} h_{3,\sigma}(0,\bold k) [ \tilde H - E + |k| ]^{-1} \big ( P_3 - \d \Gamma( k_3 ) - g A_{3,\sigma}(x',0) \big ) \Phi,
\end{split}
\end{equation}
where
\begin{equation}
\| L_1( \bold k ) \Phi \| \le \mathrm{C} |g| |k|^{-1/2} \big ( \| \Phi \| + \| x'_1 \Phi \| + \| x'_2 \Phi \| ).
\end{equation}
In passing from \eqref{eq:pull-through2} to \eqref{eq:pullthrough6} we used that 
\begin{equation}\label{eq:pullthrough7}
\left | h_{3,\sigma}(x',\bold k) - h_{3,\sigma}(0,\bold k) \right | \le \mathrm{C} |k| |x'|.
\end{equation}
Let us now note the following obvious identity:
\begin{equation}\label{eq:pullthrough8}
\frac{ \tilde H - E }{ \tilde H - E + |k| } = \frac{ \tilde E - E }{ \tilde E - E + |k| } + \frac{ |k| }{ \tilde E - E + |k| } \bigg ( \frac{ \tilde H - \tilde E }{ \tilde H - E + |k| } \bigg ).
\end{equation}
Hence, introducing \eqref{eq:tildeH-H} and \eqref{eq:pullthrough8} into \eqref{eq:pullthrough6} leads to
\begin{equation}
\begin{split}
&a( \bold k ) \Phi = L_1( \bold k )\Phi - \frac{gk_3}{2^{\frac{5}{2}}} h_{3,\sigma}(0,\bold k) [ \tilde H - E + |k| ]^{-1} \Phi \\
&- \frac{g}{2^{\frac{3}{2}}k_3} \frac{ \tilde E - E }{ \tilde E - E + |k| }  h_{3,\sigma}(0,\bold k) \Phi \\
&- \frac{g}{2^{\frac{3}{2}}k_3} \frac{ |k| }{ \tilde E - E + |k| }  h_{3,\sigma}(0,\bold k) \frac{ \tilde H - \tilde E }{ \tilde H - E + |k| } \Phi.
\end{split}
\end{equation}
We conclude the proof using again that $\|x'_j \Phi \|<\infty$.
\end{proof}
The following lemma shows in particular that if the map $P_3 \mapsto E(P_3)$ is sufficiently regular, then $\bold k \mapsto \| R(\bold k ) \Phi \|$ is in $\mathrm{L}^2( \mathbb{R}^3 \times \mathbb{Z}_2 )$, where $R( \bold k )$ denotes the operator defined in \eqref{eq:def_R} for $\sigma=0$.
\begin{lemma}\label{lm:estimate}
Let the parameters $g,\sigma,P_3$ be fixed. Assume that there exist $\gamma>0$, $P_0>0$ and a positive constant $\mathrm{C}$ independent of $\sigma \ge 0$ such that for all $|k_3| \le P_0$,
\begin{equation}\label{eq:assumption_regularity}
\left | E'_{g\sigma}(P_3+k_3) - E'_{g\sigma}(P_3) \right | \le \mathrm{C} |k_3|^{\gamma}.
\end{equation}
Then there exists a positive constant $\mathrm{C}'$, independent of $\sigma$, such that
\begin{equation}
\left \| \big ( H_{g}^\sigma(P_3-k_3) - E_{g\sigma}(P_3-k_3) \big )^{1/2} \Phi \right \| \le \mathrm{C}' |k_3|^{\frac{1+\gamma}{2}}.
\end{equation}
\end{lemma}
\begin{proof}
We use again the notations \eqref{eq:notations} and let in addition $E' = E'_{g\sigma}(P_3)$. By \eqref{eq:tildeH-H}, we have
\begin{equation}
\tilde E - E \le ( \Phi , ( \tilde H - H ) \Phi ) = - \frac{k_3}{m} \big ( \Phi , \big ( P_3 - \d \Gamma(k_3) - A_{3,\sigma}(x',0) \big ) \Phi \big ) + \frac{k_3^2}{2m}.
\end{equation}
Dividing by $-k_3$ and letting $k_3 \rightarrow 0$ (distinguishing the cases $k_3>0$ and $k_3<0$), we obtain the Feynman-Hellman formula:
\begin{equation}\label{eq:Feynman-Hellman}
E' = \frac{1}{m} \big ( \Phi , \big ( P_3 - \d \Gamma(k_3) - A_{3,\sigma}(x',0) \big ) \Phi \big ).
\end{equation}
Hence, by \eqref{eq:tildeH-H},
\begin{equation}
\begin{split}
& \left | ( \Phi , ( \tilde H - \tilde E ) \Phi ) \right | = \left | ( \Phi , ( \tilde H - H ) - ( \tilde E - E ) \Phi ) \right | \\
& \le \left | -k_3 E' - ( \tilde E - E ) \right  | + \frac{k_3^2}{2m}.
\end{split}
\end{equation}
The lemma then follows from \eqref{eq:assumption_regularity} and the mean value theorem.
\end{proof}
We are now ready to prove Theorem \ref{thm:main}: \\\\
\textsc{Proof of Theorem \ref{thm:main}}.
 Let us begin with estimating the term $ \| R_\sigma( \bold k ) \Phi_{g}^\sigma(P_3) \|$ appearing in Lemma \ref{lm:pull-through}. Recalling the notations \eqref{eq:notations}, we write
\begin{equation}
\| R_\sigma( \bold k ) \Phi \| \le \frac{ \mathrm{C} |g| }{ | k_3 | |k|^{\frac{1}{2}} } \mathds{1}_{\sigma \le |k| \le \Lambda}(k) \left \| \frac{ ( \tilde H - \tilde E )^{1/2} }{ \tilde H - E + |k| } \right \| \left \| ( \tilde H - \tilde E )^{1/2} \Phi \right \|.
\end{equation}
It follows from the Spectral Theorem and \eqref{eq:tildeE-E} that
\begin{equation}
\left \| \frac{ ( \tilde H - \tilde E )^{1/2} }{ \tilde H - E + |k| } \right \| = \sup_{r \ge 0} \left | \frac{ r^{\frac{1}{2}} }{ r + \tilde E - E + |k| } \right | \le \sup_{r\ge 0} \left | \frac{ r^{\frac{1}{2}} }{ r + |k|/4 } \right | \le \frac{ \mathrm{C} }{ |k|^{\frac{1}{2}} }.
\end{equation}
Thus, Theorem \ref{thm:regularity} together with Lemma \ref{lm:estimate} yield
\begin{equation}
\| R_\sigma( \bold k ) \Phi \| \le \frac{ \mathrm{C} |g| }{ | k_3 |^{\frac{1}{2} - \frac{\gamma}{2}} |k| } \mathds{1}_{\sigma \le |k| \le \Lambda}(k),
\end{equation}
where $\gamma = 1/4 - \delta$, and where $\delta$ in Theorem \ref{thm:regularity} is chosen such that $0 < \delta < 1/4$. Hence
\begin{equation}\label{eq:estimate2}
\int_{\mathbb{R}^3 \times \mathbb{Z}_2 } \| R_\sigma( \bold k ) \Phi \|^2 d\bold k \le \mathrm{C} g^2.
\end{equation} 

Let us now prove (i). First assume that $E'_g(P_3)=0$. In order to get the existence of a ground state for $H_g(P_3)$ our aim is to prove that $\Phi_{g}^\sigma(P_3)$ converges strongly as $\sigma \to 0$. 
Using Lemma \ref{lm:|Etau-Esigma|} (see also Remark \ref{rk:Wick}), we obtain from \eqref{eq:def_fsigma} that
\begin{equation}
| f_\sigma( \bold k ) | \le \mathrm{C} \left ( \frac{ g^2 \sigma }{ |k_3| |k|^{\frac{3}{2}} } + \frac{ |g| ( E_g(P_3-k_3) - E_g(P_3) ) }{ | k_3 | |k|^{\frac{3}{2}} } \right ) \mathds{1}_{\sigma \le |k| \le \Lambda }(k).
\end{equation}
Hence, since $E'_g(P_3)=0$ by assumption, \eqref{eq:regularity} implies
\begin{equation}
| f_\sigma( \bold k ) | \le \mathrm{C} \left ( \frac{ g^2 \sigma }{ | k_3 | |k|^{3/2} } + \frac{ |g| k_3^{\frac{1}{4}-\delta} }{ |k|^{\frac{3}{2}} } \right ) \mathds{1}_{\sigma \le |k| \le \Lambda }(k).
\end{equation}
Therefore
\begin{equation}\label{eq:estimate1}
\| f_\sigma \|_{\mathrm{L}^2( \mathbb{R}^3 \times \mathbb{Z}_2 ) } \le \mathrm{C} |g|.
\end{equation}
Combining Lemma \ref{lm:pull-through} with \eqref{eq:estimate1} and \eqref{eq:estimate2}, we obtain
\begin{equation}
( \Phi_{g}^\sigma(P_3) , \mathcal{N} \Phi_{g}^\sigma(P_3) ) = \int_{\mathbb{R}^3\times \mathbb{Z}^2} \| a (\bold k) \Phi_{g}^\sigma(P_3) \|^2 d \bold k \le \mathrm{C} g^2,
\end{equation}
where $\mathcal{N}= \d \Gamma( I )$ denotes the number operator. For a sufficiently small fixed $|g|$, the strong convergence of $\Phi_{g}^\sigma(P_3)$ as $\sigma \to 0$ is then obtained by following for instance \cite{BFS}, showing that $| ( \Phi_{g}^\sigma(P_3) , \Phi_{\rm el} \otimes \Omega ) | \ge \mathrm{C} > 0$ uniformly in $\sigma \ge 0$. Here $\Phi_{\rm el}$ denotes a normalized ground state of $h(b,V)$.

Assume next that $E'_g(P_3) \neq 0$ and let us prove that $H_g(P_3)$ does not have a ground state. By Lemmata \ref{lm:DG}, \ref{lm:pull-through} and Estimate \eqref{eq:estimate2}, it suffices to prove that $f \notin \mathrm{L}^2( \mathbb{R}^3 \times \mathbb{Z}_2 )$. The latter follows from the fact that
\begin{equation}
\left | \frac{ E_g(P_3-k_3) - E_g(P_3) }{ k_3 } \right | \ge \mathrm{C} > 0
\end{equation}
uniformly for small $k_3$ since $E'_g(P_3) \neq 0$. Hence Theorem \ref{thm:main}(i) is proven.

Let us finally prove (ii). For $\sigma > 0$, we set
\begin{equation}
\Phi^{\mathrm{ren}} = W(if_\sigma) \Phi_{g}^\sigma(P_3).
\end{equation}
Obviously $\Phi^{\mathrm{ren}}$ is a normalized ground state of $H^{\mathrm{ren}}_{g\sigma}(P_3)$.
By Lemma \ref{lm:pull-through} we have
\begin{equation*}
\begin{split}
& a( \bold k ) \Phi^{\mathrm{ren}} = W(if_\sigma) a( \bold k ) \Phi + [ a( \bold k ) , W(if_\sigma) ] \Phi \\
& = W(if_\sigma) L_\sigma(\bold k) \Phi + W(if_\sigma) R_\sigma(\bold k) \Phi + \frac{1}{\sqrt{2}} f_\sigma( \bold k ) \Phi^{\mathrm{ren}} + [ a( \bold k ) , W(if_\sigma) ] \Phi.
\end{split}
\end{equation*}
One can compute the commutator $[ a( \bold k ) , W(if_\sigma) ] = -2^{-1/2} f_\sigma( \bold k )$, so that
\begin{equation}
a( \bold k ) \Phi^{\mathrm{ren}} = W(if_\sigma) L_\sigma(\bold k) \Phi + W(if_\sigma) R_\sigma(\bold k) \Phi.
\end{equation}
Therefore, since $W(if_\sigma)$ is unitary, $\| a( \bold k ) \Phi^{\mathrm{ren}} \|$ can be estimated in the same way as $\| a ( \bold k ) \Phi \|$ (in the case $E'_g(P_3)=0$), using \eqref{eq:estimate_Lsigma} and \eqref{eq:estimate2}. This leads to the existence of a ground state for $H_g^{\mathrm{ren}}(P_3)$ and concludes the proof of Theorem \ref{thm:main}.
\hfill $\square$

\appendix

\section{Uniform regularity of the map $P_3 \mapsto E_{g\sigma}(P_3)$}

In this appendix we shall prove Theorem \ref{thm:regularity}. The structure follows \cite{Pizzo} and \cite{CFP}: First, we give a simple proof of the existence of a spectral gap for the infrared cutoff Hamiltonian $H_{g}^\sigma(P_3)$, considered as an operator on the space of photons of energies $\ge \sigma$. Our proof is based on the min-max principle. Then we establish \eqref{eq:regularity} by adapting \cite{Pizzo,CFP} (see also \cite{BFP}). In comparison to \cite{CFP}, the main technical difference comes from the terms in $H_g(P_3)$ containing the interaction between the electronic variables $x'_j$ and the quantized electromagnetic field. This shall be handled in Lemma \ref{lm:2nd_step} below thanks to the exponential decay of $\Phi_g^\sigma(P_3)$ in $x'_j$.  

In some parts of our presentation, we shall only sketch the proof, emphasizing the differences that we have to include, and referring otherwise to \cite{Pizzo}, \cite{BFP}, or \cite{CFP}.

Let us begin with some definitions and notations. Henceforth we remove the subindex $g$ to simplify the notations, and for $\sigma \ge 0$, we replace $H^{\sigma}(P_3)$ by its Wick-ordered version $H^\sigma(P_3) - \frac{g^2}{2m}( \Lambda^2 - \sigma^2)$ (which we still denote by $H^\sigma(P_3)$). Note that this shall not affect our discussion below on the regularity of the ground state energy since the two operators only differ by a constant. We decompose
\begin{equation}
H^\sigma(P_3) = h_0(P_3) + H_I^\sigma(P_3),
\end{equation}
where
\begin{equation}
h_0(P_3) = h(b,V) \otimes \mathds{1} + \mathds{1} \otimes \left [ \frac{1}{2m} \Big ( P_3 - d\Gamma(k_3) \Big )^2 + H_f \right ],
\end{equation}
and
\begin{equation}
\begin{split}
H_I^\sigma(P_3) =& - \frac{g}{m} \sum_{j=1,2} \bigg ( A_{j,\sigma} (x',0) \Big ( p_j - e a_j(x') \Big ) + \frac{g^2}{2m} A_{j,\sigma}(x',0)^2 \bigg ) \\
& - \frac{g}{2m} A_{3,\sigma} (x',0)\Big ( P_3 - d \Gamma(k_3) \Big ) - \frac{g}{2m} \Big ( P_3 - d \Gamma(k_3) \Big ) A_{3,\sigma} (x',0)  \\
& + \frac{g^2}{2m} A_{3,\sigma} (x',0)^2 - \frac{g}{2m} \sigma \cdot B_\sigma(x',0) - \frac{g^2}{2m} ( \Lambda^2 - \sigma^2 ).
\end{split}
\end{equation}
Let $\Phi_{\rm el}$ denote a normalized ground state of $h(b,V)$. For any $|P_3| < m$, one can easily check that $\Phi_{\rm el} \otimes \Omega$ is a ground state of $h_0(P_3)$, with ground state energy $e_0(P_3) = e_0 + P_3^2/2m$. Note that for $\tau \le \sigma$, we have
\begin{equation}\label{eq:Htausigma1}
\begin{split}
& H^\tau(P_3) - H^\sigma(P_3) \phantom{\sum} \\
&= - \frac{g}{m} \sum_{j=1,2} A_{j,\tau}^\sigma (x',0) \Big ( p_j - e a_j( x' ) - g A_{j,\sigma}(x',0) \Big ) - \frac{g^2}{2m} ( \sigma^2 - \tau^2 ) \\
&+ \frac{g^2}{2m} A_\tau^\sigma (x',0)^2 - \frac{g}{2m} A_{3,\tau}^\sigma (x',0) \Big ( P_3 - d\Gamma(k_3) - g A_{3,\sigma} (x',0) \Big ) \phantom{\sum_j}\\
&- \frac{g}{2m} \Big ( P_3 - d\Gamma(k_3) - g A_{3,\sigma} (x',0) \Big ) A_{3,\tau}^\sigma (x',0) - \frac{g}{2m} \sigma \cdot B_\tau^\sigma(x',0), 
\end{split}
\end{equation} 
where 
\begin{equation}
A_\tau^\sigma(x',0) = \frac{1}{\sqrt{2\pi}} \int \frac{\epsilon^\lambda(k)}{ |k|^{1/2} } \rho_\tau^\sigma(k) \left [Êe^{ - i k' \cdot x' } a^*_\lambda(k) + e^{ i k' \cdot x' } a_\lambda(k) \right ] d \bold k,
\end{equation}
and likewise for $B_\tau^\sigma(x',0)$.
Let $\mathcal{H}_\sigma = \mathrm{L}^2( \mathbb{R}^2 ; \mathbb{C}^2 ) \otimes \mathcal{F}_\sigma$, where $\mathcal{F}_\sigma$ denotes the symmetric Fock space over $\mathrm{L}^2( \{ \bold k \in \mathbb{R}^3 \times \mathbb{Z}_2 , |k| \ge \sigma \} )$. The restriction of $H^\sigma(P_3)$ to $\mathcal{H}_\sigma$ is denoted by $H_\sigma(P_3)$:
\begin{equation}
H_\sigma(P_3) = H^\sigma(P_3) |_{\mathcal{H}_\sigma},
\end{equation}
and, similarly,
\begin{equation}
h_{0,\sigma}(P_3) = h_0(P_3) |_{\mathcal{H}_\sigma} \quad , \quad H_{I,\sigma}(P_3) = H_I^\sigma(P_3) |_{\mathcal{H}_\sigma}.
\end{equation}
Let $\Omega_\sigma$ be the vacuum in $\mathcal{F}_\sigma$. Then for $|P_3| < m$, $\Phi_{\rm el} \otimes \Omega_\sigma$ is a ground state of $h_{0,\sigma}( P_3 )$ with ground state energy $e_0(P_3)$, and
\begin{equation}\label{eq:gapH0}
\mathrm{Gap}( h_{0,\sigma}( P_3 ) ) \ge ( 1 - \frac{|P_3|}{m} ) \sigma,
\end{equation}
where $\mathrm{Gap}( H ) = \inf ( \sigma ( H ) \setminus \{ E(H) \} ) - \inf ( \sigma (H) )$ for any self-adjoint and semi-bounded operator $H$ with ground state energy $E(H)$. 
We also define
\begin{equation}\label{eq:Htausigma2}
H_\tau^\sigma(P_3) = \left ( H^\tau(P_3) - H^\sigma(P_3) \right ) |_{\mathcal{H}_\tau}.
\end{equation}
The symmetric Fock space over $\mathrm{L}^2( \{ \bold k \in \mathbb{R}^3 \times \mathbb{Z}_2 , \tau \le |k| \le \sigma \} )$ is denoted by $\mathcal{F}_\tau^\sigma$. 
Note that there exists a unitary operator $ \mathcal{V} : \mathcal{H}_{\tau} \rightarrow \mathcal{ H }_{ \sigma } \otimes \mathcal{F}_{ \tau }^{ \sigma }$. We shall identify $\mathcal{H}_\tau$ and $\mathcal{H}_\sigma \otimes \mathcal{F}_\tau^\sigma$ in the sequel in order to simplify the notations. We let $\Omega_\tau^\sigma$ be the vacuum in $\mathcal{F}_\tau^\sigma$.

\subsection{Existence of a spectral gap}\label{subsection:gap}
\begin{lemma}\label{lm:gap1}
There exist $g_0>0$, $\sigma_0>0$ and $P_0>0$ such that the following holds: Let $|g| \le g_0$, $0 \le \sigma \le \sigma_0$ and $|P_3| \le P_0$ be such that $H_\sigma( P_3 )$ has a normalized ground state $\Phi_\sigma( P_3 )$ and $\mathrm{Gap}( H_\sigma ( P_3 ) ) \ge \gamma \sigma$ for some $\gamma > 0$. Then for all $0 \le \tau \le \sigma$, $\Phi_\sigma(P_3) \otimes \Omega_\tau^\sigma$ is a normalized ground state of $  H^\sigma( P_3 ) |_{\mathcal{H}_\tau}  $, and
\begin{equation}\label{eq:gap_Hsigma|tau}
\mathrm{Gap} ( H^\sigma( P_3 ) |_{\mathcal{H}_\tau} ) \ge \min ( \gamma \sigma ,  \tau / 4 ).
\end{equation}
\end{lemma}
\begin{proof}
To simplify the notations, let us remove the dependence on $P_3$ throughout the proof. First, one can readily check that $\Phi_\sigma \otimes \Omega_\tau^\sigma$ is an eigenstate of $  H^\sigma |_{\mathcal{H}_\tau}  $ associated with the eigenvalue $E_\sigma$. For any $v$ we let $[v]$ and $[v]^\perp$ denote respectively the subspace spanned by $v$ and its orthogonal complement. We write
\begin{equation*}
\begin{split}
& \inf_{\Phi \in [\Phi_\sigma \otimes \Omega_\tau^\sigma]^\perp, \| \Phi \|=1 } ( \Phi ,   H^\sigma |_{\mathcal{H}_\tau}   \Phi ) \\
& \ge \min \Big ( \inf_{\Phi \in [\Phi_\sigma]^\perp \otimes [\Omega_\tau^\sigma], \| \Phi \|=1 } ( \Phi ,   H^\sigma |_{\mathcal{H}_\tau}   \Phi ) ,  \inf_{\Phi \in \mathcal{H}_\sigma \otimes [\Omega_\tau^\sigma]^\perp, \| \Phi \|=1 } ( \Phi ,   H^\sigma |_{\mathcal{H}_\tau}   \Phi ) \Big ).
\end{split}
\end{equation*}
The assumption $\mathrm{Gap}( H_\sigma ) \ge \gamma \sigma$ implies
\begin{equation*}
\inf_{\Phi \in [\Phi_\sigma]^\perp \otimes [\Omega_\tau^\sigma], \| \Phi \|=1 } ( \Phi ,   H^\sigma |_{\mathcal{H}_\tau}   \Phi ) \ge E_\sigma + \gamma \sigma.
\end{equation*}
On the other hand, using that the number operator $\int_{ \tau \le |k| \le \sigma } a^*( \bold k ) a( \bold k ) d \bold k$ commutes with $H^\sigma |_{\mathcal{H}_\tau}$, one can prove as in \cite{Pizzo} that
\begin{equation*}
\inf_{\Phi \in \mathcal{H}_\sigma \otimes [\Omega_\tau^\sigma]^\perp, \| \Phi \|=1 } ( \Phi ,   H^\sigma |_{\mathcal{H}_\tau}   \Phi ) \ge \inf_{\tau \le |k| \le \sigma} ( E_\sigma(P_3-k_3) - E_\sigma(P_3) + |k| ).
\end{equation*}
We conclude the proof thanks to \eqref{eq:tildeE-E}
\end{proof}
\begin{corollary}\label{cor:Etau-Esigma}
Under the conditions of Lemma \ref{lm:gap1}, for all $0 \le \tau \le \sigma$,
\begin{equation}\label{eq:Etau-Esigma}
E_\tau(P_3) \le E_\sigma(P_3) \le e_0(P_3).
\end{equation}
\end{corollary}
\begin{proof}
It follows from Lemma \ref{lm:gap1} that 
\begin{equation}
\begin{split}
E_\tau(P_3) & \le ( \Phi_\sigma(P_3) \otimes \Omega_\tau^\sigma , H_\tau(P_3) \Phi_\sigma(P_3) \otimes \Omega_\tau^\sigma ) \\
&= ( \Phi_\sigma(P_3) \otimes \Omega_\tau^\sigma , H^\sigma(P_3) |_{\mathcal{H}_\tau} \Phi_\sigma(P_3) \otimes \Omega_\tau^\sigma ) = E_\sigma(P_3).
\end{split}
\end{equation}
Hence the first inequality in \eqref{eq:Etau-Esigma} is proven. To prove the second one, it suffices to write similarly
\begin{equation}
\begin{split}
E_\sigma(P_3) & \le ( \Phi_{\rm el} \otimes \Omega_\sigma , H_\sigma(P_3) \Phi_{\rm el} \otimes \Omega_\sigma ) \\
&= ( \Phi_{\rm el} \otimes \Omega_\sigma , h_{0,\sigma}(P_3) \Phi_{\rm el} \otimes \Omega_\sigma ) = e_0(P_3).
\end{split}
\end{equation}
\end{proof}

We shall establish the existence of a spectral gap of order $O(\sigma)$ above the bottom of the spectrum of $H_{\sigma}(P_3)$ by induction. More precisely, let \textbf{Gap}($\sigma$) denote the assertion
$$
\mathbf{Gap}(\sigma)
\left \{ 
\begin{array}{cl}
 &\mathrm{(i)} \phantom{i} \quad E_{\sigma}(P_3) \text{ is a simple eigenvalue of } H_{\sigma}(P_3), \vspace{0,1cm} \\
 &\mathrm{(ii)} \quad \mathrm{Gap}( H_\sigma( P_3 ) ) \ge \sigma/8.
\end{array}
\right.
$$
We shall prove
\begin{prop}\label{prop:gap}
There exists $g_0>0$, $\sigma_0>0$ and $P_0>0$ such that, for all $|g| \le g_0$, $0 < \sigma \le \sigma_0$ and $|P_3| \le P_0$, the assertion $\mathbf{Gap}(\sigma)$ above holds.
\end{prop}
Let us begin with two preliminary useful estimates:
\begin{lemma}
Fix the parameters $g$, $\sigma$ and $P_3$ such that $0 \le |g| \le g_0$, $0\le \sigma \le \sigma_0$ and $0\le |P_3| \le P_0$, for some sufficiently small small $g_0$, $\sigma_0$ and $P_0$. For any $0 < \rho < 1$,
\begin{equation}\label{eq:main_estimate_root}
\begin{split}
&\left \| \left [ h_{0,\sigma}(P_3) - e_0(P_3) + \rho \right ]^{-1/2} H_{I,\sigma}(P_3) \left [ h_{0,\sigma}(P_3) - e_0(P_3) + \rho \right ]^{-1/2} \right \| \\
& \le \mathrm{C} |g| \rho^{-1/2},
\end{split}
\end{equation}
where $\mathrm{C}$ is a positive constant (depending only on $\Lambda$). Likewise,
\begin{equation}\label{eq:main_estimate}
\begin{split}
&\left \| \left [ H^{\sigma}(P_3) |_{\mathcal{H}_\tau} - E_\sigma(P_3) + \rho \right ]^{-1/2} H_\tau^\sigma(P_3) \left [ H^{\sigma}(P_3) |_{\mathcal{H}_\tau}  - E_\sigma(P_3) + \rho \right ]^{-1/2} \right \| \\
&\leq \mathrm{C} |g| \sigma^{1/2} \rho^{-1/2}.
\end{split}
\end{equation}
\end{lemma}
\begin{proof}
Let us prove \eqref{eq:main_estimate}, Estimate \eqref{eq:main_estimate_root} would follow similarly. 
We introduce the expression of $H_\tau^\sigma(P_3)$ given by \eqref{eq:Htausigma1} and \eqref{eq:Htausigma2} and estimate each term separately. 
Consider for instance
\begin{equation}\label{eq:lemma_main_estimate1}
\begin{split}
|g| \bigg \| &\Big [ H^{\sigma}(P_3) |_{\mathcal{H}_\tau} - E_\sigma(P_3) + \rho \Big ]^{-1/2} \int_{ \tau \le |k| \le \sigma } \frac{ \epsilon^{(3)}_\lambda(k) }{ |k|^{1/2} }e^{ik'\cdot x'} a^*(\bold k) d \bold k \\
&\quad \Big ( P_3 - d\Gamma(k_3) +g A_{3,\sigma}(x',0) \Big ) \Big [ H^{\sigma}(P_3) |_{\mathcal{H}_\tau}  - E_\sigma(P_3) + \rho \Big ]^{-1/2} \bigg \|.
\end{split}
\end{equation}
Using that
\begin{equation}
\left \| \Big ( P_3 - d\Gamma(k_3) +g A_{3,\sigma}(x',0) \Big ) \Big [ H^{\sigma}(P_3) |_{\mathcal{H}_\tau}  - E_\sigma(P_3) + \rho \Big ]^{-1/2} \right \| \le \mathrm{C} \rho^{-1/2},
\end{equation}
we get
\begin{equation*}
\eqref{eq:lemma_main_estimate1} \le \mathrm{C} |g| \rho^{-1/2} \bigg \| \Big [ H^{\sigma}(P_3) |_{\mathcal{H}_\tau} - E_\sigma(P_3) + \rho \Big ]^{-1/2} \int_{ \tau \le |k| \le \sigma } \frac{ \epsilon^{(3)}_\lambda(k) }{ |k|^{1/2} }e^{ik'\cdot x'} a^*(\bold k) d \bold k \bigg \|.
\end{equation*}
Moreover, for any $\Phi \in D(H^\sigma(P_3) |_{\mathcal{H}_\tau})$,
\begin{equation*}
\begin{split}
& \bigg \| \Big [ H^{\sigma}(P_3) |_{\mathcal{H}_\tau} - E_\sigma(P_3) + \rho \Big ]^{-1/2} \int_{ \tau \le |k| \le \sigma } \frac{ \epsilon^{(3)}_\lambda(k) }{ |k|^{1/2} }e^{ik'\cdot x'} a^*( \bold k) d \bold k \Phi \bigg \|^2 \\
& \le \int_{ \tau \le |k|,|\tilde{k}| \le \sigma }  \frac{ \mathrm{C} }{ |k|^{1/2} |\tilde{k}|^{1/2} } \left | \bigg ( \Phi , a(\bold k) \Big [ H^{\sigma}(P_3) |_{\mathcal{H}_\tau} - E_\sigma(P_3) + \rho \Big ]^{-1} a^*(\tilde{\bold k}) \Phi \bigg )  \right | d \bold k d \tilde{\bold k} .
\end{split}
\end{equation*}
Now, for any $\bold k$ such that $\tau \le |k| \le \sigma$, we have the pull-through formula
\begin{equation}
a(\bold k) H^{\sigma}(P_3) |_{\mathcal{H}_\tau} = \Big [ H^{\sigma}(P_3 - k_3)|_{\mathcal{H}_\tau} + |k| \Big ] a(\bold k),
\end{equation}
since $a(\bold k)$ commutes with $A_{\sigma}(x',0)$. Hence
\begin{equation*}
\begin{split}
& \bigg ( \Phi , a(\bold k) \Big [ H^{\sigma}(P_3) |_{\mathcal{H}_\tau} - E_\sigma(P_3) + \rho \Big ]^{-1} a^* (\tilde{\bold k}) \Phi \bigg ) \\
&= \delta( \bold k - \tilde{\bold k} ) \bigg ( \Phi , \Big [ H^{\sigma}(P_3 - k_3) |_{\mathcal{H}_\tau} - E_\sigma(P_3) + |k| +  \rho \Big ]^{-1}  \Phi \bigg ) \\
&\quad + \bigg ( a( \tilde{\bold k} ) \Phi , \Big [ H^{\sigma}(P_3 - k_3 - \tilde{k}_3) |_{\mathcal{H}_\tau} - E_\sigma(P_3) + |k| + |\tilde{k}| +  \rho \Big ]^{-1} a(\bold k)  \Phi \bigg ). 
\end{split}
\end{equation*}
Using that $H^{\sigma}(P_3 - k_3) |_{\mathcal{H}_\tau} - E_\sigma(P_3) + |k| \ge |k|/4$ for any $k$ sufficiently small (see \eqref{eq:tildeE-E}), we get
\begin{equation*}
\begin{split}
& \Big \| \Big [ H^{\sigma}(P_3 - k_3) |_{\mathcal{H}_\tau} - E_\sigma(P_3) + |k| + \rho \Big ]^{-1} \Big \| \le \frac{ \mathrm{C} }{ |k| }. 
\end{split}
\end{equation*}
Let $H_{f,\tau}^\sigma = \int_{ \tau \le |k| \le \sigma } |k| a^*(\bold k) a(\bold k) d \bold k$. As in \cite[Lemma 1.1]{Pizzo}, it follows from the proof of Lemma \ref{lm:gap1} that $H_{f,\tau}^\sigma \le \mathrm{C} ( H^\sigma(P_3) |_{\mathcal{H}_\tau} - E_\sigma(P_3 ))$ for any $P_3$ sufficiently small. This yields
\begin{equation*}
\left \| \big [ H_{f,\tau}^\sigma + |k| + |\tilde k| \big ] \Big [ H^{\sigma}(P_3 - k_3 - \tilde k_3) |_{\mathcal{H}_\tau} - E_\sigma(P_3) + |k| + |\tilde k| + \rho \Big ]^{-1} \right \| \le \mathrm{C}.
\end{equation*}
Thus, combining the previous estimates we obtain
\begin{equation*}
\begin{split}
& \bigg \| \Big [ H^{\sigma}(P_3) |_{\mathcal{H}_\tau} - E_\sigma(P_3) + \rho \Big ]^{-1/2} \int_{ \tau \le |k| \le \sigma } \frac{ \epsilon^{(3)}_\lambda(k) }{ |k|^{1/2} }e^{ik'\cdot x'} a^*(\bold k) d\bold k \Phi \bigg \|^2 \\
& \le \mathrm{C} \int_{ \tau \le |k| \le \sigma } \frac{ d \bold k }{ |k|^2 } + \mathrm{C} \left [ \int_{ \tau \le |k| \le \sigma } \frac{ d \bold k }{ |k|^{\frac{1}{2}} }  \Big \| [H_{f,\tau}^\sigma + |k|]^{-1/2} a(\bold k) \Phi \Big \| \right ]^2 \le \mathrm{C} \sigma.
\end{split}
\end{equation*}
Since $D(H^\sigma(P_3) |_{\mathcal{H}_\tau})$ is dense in $\mathcal{H}_\tau$, the result is proven as for the term we have chosen to consider, that is $\eqref{eq:lemma_main_estimate1} \le \mathrm{C} |g| \sigma^{1/2} \rho^{-1/2}$. Since the other terms in the expression of $H^\sigma_\tau$ given by \eqref{eq:Htausigma1} can be treated in the same way, the lemma is established.
\end{proof}
The next lemma corresponds to the root in the induction procedure leading to the proof of Proposition \ref{prop:gap}.
\begin{lemma}\label{lm:root}
There exist $g_0>0$, $\sigma_0>0$, $P_0>0$ and a positive constant $\mathrm{C}_0$ such that for all $|g| \le g_0$ and $|P_3| \le P_0$, for all  $\sigma$ such that $\mathrm{C}_0 g^2 \le \sigma \le \sigma_0$, the assertion $\mathbf{Gap}(\sigma)$ holds.
\end{lemma}
\begin{proof}
To simplify the notations, we write $H_\sigma$ for $H_\sigma(P_3)$, $E_\sigma$ for $E_\sigma(P_3)$, and similarly for other quantities depending on $P_3$. Let $\mu_{\sigma}$ denote the first point above $E_{\sigma}$ in the spectrum of $H_{\sigma}$. By the min-max principle,
\begin{equation}\label{eq:min-max}
\begin{split}
\mu_{\sigma} &\ge \inf_{ \psi \in [ \Phi_{\rm el} \otimes \Omega_\sigma ]^\perp , \| \psi \| = 1} ( \psi , H_{\sigma} \psi ),
\end{split}
\end{equation}
where $[v]^\perp$ denotes the orthogonal complement of the vector space spanned by $v$. It follows from \eqref{eq:main_estimate_root} that for any $\psi \in [ \Phi_{\rm el} \otimes \Omega_\sigma ]^\perp$, $\| \psi \| = 1$, and any $\rho > 0$,
\begin{equation}
\begin{split}
(\psi , H_{\sigma} \psi ) &\ge ( \psi , H_{ 0,\sigma } \psi ) - \mathrm{C} |g| \rho^{-1/2} ( \psi , [ h_{0,\sigma} - e_0(P_3) + \rho ] \psi ) \\
&\ge \left ( 1 - \mathrm{C} |g| \rho^{-1/2} \right ) ( \psi , H_{ 0,\sigma } \psi ) + \mathrm{C} |g| \rho^{-1/2} e_0(P_3) - \mathrm{C} |g| \rho^{1/2}.
\end{split}
\end{equation}
By \eqref{eq:gapH0}, for any $\psi \in [ \Phi_{ \mathrm{el} } \otimes \Omega_\sigma ]^\perp$, $(\psi , h_{0,\sigma} \psi ) \ge e_0(P_3) + (1-|P_3|/m) \sigma$ provided that $\sigma_0$ is chosen sufficiently small. Hence for any $\rho$  such that $\rho^{1/2} > \mathrm{C} |g|$,
\begin{equation}
\begin{split}
(\psi , H_{\sigma} \psi ) &\ge e_0(P_3) + \left ( 1 - \mathrm{C} |g| \rho^{-1/2} \right ) \left ( 1 - \frac{|P_3|}{m} \right ) \sigma - \mathrm{C} |g| \rho^{1/2}.
\end{split}
\end{equation}
Choosing $\rho^{1/2} = 4 \mathrm{C} |g|$ and $P_0$ sufficiently small, by Corollary \ref{cor:Etau-Esigma}, we obtain
\begin{equation}
\begin{split}
(\psi , H_{\sigma} \psi ) &\ge E_{\sigma} + \frac{3}{4} \left ( 1 - \frac{ |P_3| }{m} \right ) \sigma - 4 \mathrm{C}^2 g^2 \\
&\ge E_\sigma + \frac{1}{2} \sigma - 4 \mathrm{C}^2 g^2.
\end{split}
\end{equation}
Together with \eqref{eq:min-max}, this leads to the statement of the lemma provided that the constant $\mathrm{C}_0$ is chosen such that $\mathrm{C}_0 > 32 \mathrm{C}^2 / 3$.
\end{proof}
The following lemma corresponds to the induction step of the induction process in the proof of Proposition \ref{prop:gap}.
\begin{lemma}\label{lm:induction}
There exists $g_0>0$, $\sigma_0>0$ and $P_0>0$ such that for all $|g| \le g_0$ and $|P_3| \le P_0$, for all $\sigma$ such that $0 < \sigma \le \sigma_0$,
\begin{equation*}
\mathbf{Gap}(\sigma) \Rightarrow \mathbf{Gap}(\sigma/2).
\end{equation*}
\end{lemma}
\begin{proof}
Again, throughout the proof, we drop the dependence on $P_3$ in all the considered quantities. Let $\mathbf{Gap}(\sigma)$ be satisfied for some $0 < \sigma$, let $\Phi_\sigma$ be a ground state of $H_{\sigma}$, and let $\tau = \sigma/2$. As in the proof of Lemma \ref{lm:root}, let $\mu_{\tau}$ denote the first point above $E_{\tau}$ in the spectrum of $H_{\tau}$. By the min-max principle,
\begin{equation}
\begin{split}
\mu_{\tau} &\ge \inf_{ \psi \in [ \Phi_\sigma \otimes \Omega_\tau^\sigma ]^\perp , \| \psi \| = 1} ( \psi ,   H_{\tau}   \psi ),
\end{split}
\end{equation}
where $\Omega_\tau^\sigma$ is the vacuum in $\mathcal{F}_\tau^\sigma$ and where $[ \Phi_\sigma \otimes \Omega_\tau^\sigma ]^\perp$ denotes the orthogonal complement of the vector space spanned by $\Phi_\sigma \otimes \Omega_\tau^\sigma$ in $\mathcal{H}_\sigma \otimes \mathcal{F}_\tau^\sigma$. It follows from \eqref{eq:main_estimate} that for any $\rho>0$,
\begin{equation*}
\begin{split}
&( \psi ,   H_{\tau}   \psi ) \ge ( \psi ,   H^\sigma |_{\mathcal{H}_\tau}   \psi ) + ( \psi ,   H_\tau^\sigma   \psi ) \\ 
&\ge \left [ 1 - \mathrm{C} |g| \sigma^{1/2} \rho^{-1/2} \right ] ( \psi ,   H^\sigma |_{\mathcal{H}_\tau}   \psi ) + \mathrm{C} |g| \sigma^{1/2} \rho^{-1/2} E_{\sigma} - \mathrm{C} |g| \sigma^{1/2} \rho^{1/2}.
\end{split}
\end{equation*}
Next, from $\mathbf{Gap}(\sigma)$ and Property \eqref{eq:gap_Hsigma|tau}, since $\tau = \sigma/2$, we obtain that for any $\psi$ in $[ \Phi_\sigma \otimes \Omega_\tau^\sigma ]^\perp$, $\| \psi \|=1$,
\begin{equation}
( \psi ,   H^\sigma |_{\mathcal{H}_\tau}   \psi ) \ge E_{\sigma} + \min \left ( \frac{\sigma}{8} , \frac{\tau}{4} \right ) \ge E_\sigma + \sigma/8,
\end{equation}
provided that $|g|$ is sufficiently small. Hence for any $\rho>0$ such that $\rho^{1/2} > \mathrm{C} |g| \sigma^{1/2}$,
\begin{equation}
\begin{split}
( \psi ,   H_{\tau}   \psi ) &\ge E_{\sigma} + \left [ 1 - \mathrm{C} |g| \sigma^{1/2} \rho^{-1/2} \right ] \frac{ \sigma }{ 8 } - \mathrm{C} |g| \sigma^{1/2} \rho^{1/2}.
\end{split}
\end{equation}
Choosing $\rho^{1/2} = 4 \mathrm{C} |g| \sigma^{1/2}$, by Corollary \ref{cor:Etau-Esigma}, we get
\begin{equation}
\begin{split}
( \psi ,   H_{\tau}   \psi ) &\ge E_{\sigma} + \frac{3}{32} \sigma - 4 \mathrm{C}^2 g^2 \sigma \ge E_{\tau} + \frac{3}{16} \tau - 8 \mathrm{C}^2 g^2 \tau.
\end{split}
\end{equation}
Hence $\mu_{\tau} \ge E_{\tau} + \tau/8$ provided that $|g| \le ( 8 \mathrm{C} )^{-1}$, which proves the lemma.
\end{proof}
\noindent \textsc{Proof of Proposition \ref{prop:gap}}
As mentioned above, Proposition \ref{prop:gap} easily follows from Lemmata \ref{lm:root} and \ref{lm:induction}, and an induction argument.
\hfill $\square$ \\

Let us conclude this Subsection with a bound on the difference $| E_\tau - E_\sigma |$.
\begin{lemma}\label{lm:|Etau-Esigma|}
Under the conditions of Proposition \ref{prop:gap}, there exists a positive constant $\mathrm{C}$ such that for all $0 \le \tau \le \sigma \le \sigma_0$,
\begin{equation}
| E_\tau(P_3) - E_\sigma(P_3) | \le \mathrm{C} |g| \sigma.
\end{equation}
\end{lemma}
\begin{proof}
By Corollary \ref{cor:Etau-Esigma}, we already have $E_\tau(P_3) \le E_\sigma(P_3)$. The inequality $E_\sigma(P_3) \le E_\tau(P_3) + \mathrm{C} |g| \sigma$ follows similarly, using \eqref{eq:main_estimate} and a variational argument.
\end{proof}
\begin{remark}\label{rk:Wick}
Lemma \ref{lm:|Etau-Esigma|} remains true if the operators under consideration are not Wick-ordered. More precisely in this case we have
\begin{equation}
E_\tau (P_3) \le E_\sigma(P_3) + \mathrm{C} g^2 \sigma \le E_\tau(P_3) + \mathrm{C} |g| \sigma.
\end{equation}
\end{remark}

\subsection{Proof of Theorem \ref{thm:regularity}}

The key property used in the proof of Theorem \ref{thm:regularity} lies in the estimate of $| E'_\tau( P_3 ) - E'_\sigma( P_3 )|$ for $\tau \le \sigma$.
\begin{prop}\label{prop:E'tau-E'sigma}
There exits $g_0>0$, $\sigma_0>0$ and $P_0>0$ such that for all $0 < |g| \le g_0$ and $|P_3| \le P_0$, for all $\sigma,\tau>0$ such that $\tau \le \sigma \le \sigma_0$, for all $\delta>0$,
\begin{equation*}
| E'_\tau( P_3 ) - E'_\sigma( P_3 )| \le \mathrm{C}_\delta \sigma^{1/2-\delta},
\end{equation*}
where $\mathrm{C}_\delta$ is a positive constant depending only on $\delta$.
\end{prop}
We shall divide the main part of the proof of Proposition \ref{prop:E'tau-E'sigma} into two lemmata. Let us begin with some definitions and notations. For $\sigma >0$ and $\rho \ge 0$, we define the function $g_{\sigma,\rho} \in \mathrm{L}^2( \mathbb{R}^3 \times \mathbb{Z}_2 )$ by
\begin{equation*}
g_{\sigma,\rho}( \bold k ) = g \mathds{1}_{ \sigma \le |k| \le \Lambda }( k ) \frac{ \epsilon^3_\lambda (k) }{ \sqrt{2\pi} |k|^{1/2} } \frac{ \rho }{ |k| - k_3 \rho }.
\end{equation*}
Depending on the context, the Weyl operator $W( ig_{\sigma,\rho} )$ will represent an operator on $\mathcal{H}_\sigma$, $\mathcal{H}_\tau$ (for $\tau \le \sigma$), or $\mathcal{H}$. 

From now on, to simplify the notations, we drop the dependence on $P_3$ everywhere unless a confusion may arise. For $g$, $\sigma$ and $P_3$ as in Proposition \ref{prop:gap}, let $\Phi_\sigma$ denote a normalized ground state of $H_\sigma$. Define
\begin{equation*}
H^{\rm ren}_{\sigma,\rho} = W( ig_{\sigma,\rho} ) H_\sigma W ( ig_{\sigma,\rho} )^* , \quad \Phi^{\rm ren}_{\sigma,\rho} = W( ig_{\sigma,\rho} ) \Phi_\sigma,
\end{equation*}
and let $P^{\rm ren}_{\sigma,\rho}$ be the orthogonal projection onto the vector space spanned by $\Phi^{\rm ren}_{\sigma,\rho}$. Note that $\Phi^{\rm ren}_{\sigma,\rho}$ is a normalized, non-degenerate ground state of $H^{\rm ren}_{\sigma,\rho}$, associated with the ground state energy $E_\sigma$. 
Recall that, by Lemma \ref{lm:gap1}, $  \left [ \Phi_\sigma \otimes \Omega_\tau^\sigma \right ]$ is a ground state of $H^\sigma |_{\mathcal{H}_\tau}$. We set
\begin{equation*}
H^{\mathrm{ren}}_{\sigma,\rho,\tau} = W( ig_{\sigma,\rho} ) H^\sigma |_{\mathcal{H}_\tau} W ( ig_{\sigma,\rho} )^* , \quad \Phi^{\rm ren}_{\sigma,\rho,\tau} = W( ig_{\sigma,\rho} )   [ \Phi_\sigma \otimes \Omega_\tau^\sigma ],
\end{equation*}
and the projection onto the vector space spanned by $\Phi^{\rm ren}_{\sigma,\rho,\tau}$ is denoted by $P^{\mathrm{ren}}_{\sigma,\rho,\tau}$. Since $W(ig_{\sigma,\rho} )= e^{i\Phi(ig_{\sigma,\rho}) \otimes \mathds{1} }$, it can be seen that $\Phi^{\rm ren}_{\sigma,\rho,\tau} =   [ W( ig_{\sigma,\rho} ) \Phi_\sigma ] \otimes \Omega_\tau^\sigma =   \Phi^{\mathrm{ren}}_{\sigma,\rho} \otimes \Omega_\tau^\sigma$. 
\begin{lemma}\label{lm:E'_GS}
There exists $g_0>0$, $\sigma_0>0$ and $P_0>0$ such that for all $0<|g|\le g_0$ and $|P_3| \le P_0$, for all $\sigma,\tau>0$ such that $\tau \le \sigma \le \sigma_0$,
\begin{equation}\label{eq:lm_E'_GS}
\left | E'_\sigma - E'_\tau \right | \le \mathrm{C} \left [ \left \| P^{\rm ren}_{\sigma,E'_\sigma,\tau} - P^{\rm ren}_{\tau,E'_\sigma} \right \| + g^2 \sigma \right ],
\end{equation}
where $\mathrm{C}$ is a positive constant.
\end{lemma}
\begin{proof}
By the Feynman-Hellman formula (see \eqref{eq:Feynman-Hellman}),
\begin{equation}\label{eq:Feynman-Hellman2}
E'_\sigma = \frac{1}{m} \left ( \Phi_\sigma , \left [ P_3 - d\Gamma (k_3) - g A_{3,\sigma}( x',0) \right ] \Phi_\sigma \right )_{\mathcal{H}_\sigma}.
\end{equation}
It follows from \eqref{eq:Feynman-Hellman2} and commutation relations with $W(ig_{\sigma,E'_\sigma})$ that
\begin{equation}
\begin{split}
E'_\sigma = &\frac{1}{m} \left ( \Phi^{\mathrm{ren}}_{\sigma,E'_\sigma} , \left [ P_3 - d\Gamma (k_3) - \Phi ( k_3 g_{\sigma,E'_\sigma} ) - \frac{1}{2} ( k_3 g_{\sigma,E'_\sigma} , g_{\sigma,E'_\sigma} ) \right.\right. \\
&\left.\left. \phantom{\frac{1}{2} W( ig_{\sigma,E'_\sigma} ) \Phi_\sigma ,} - g A_{3,\sigma}( x',0) + g \Re( h_{3,\sigma}(x') , g_{\sigma,E'_\sigma} ) \right ] \Phi^{\mathrm{ren}}_{\sigma,E'_\sigma} \right )_{\mathcal{H}_\sigma},
\end{split}
\end{equation}
Consequently, for $\tau \le \sigma$, we can write
\begin{equation}\label{eq:E'sigma}
\begin{split}
E'_\sigma = &\frac{1}{m} \left ( \Phi^{\mathrm{ren}}_{\sigma,E'_\sigma,\tau} , \left [ P_3 - d\Gamma (k_3) - \Phi ( k_3 g_{\tau,E'_\sigma} ) - \frac{1}{2} ( k_3 g_{\sigma,E'_\sigma} , g_{\sigma,E'_\sigma} ) \right.\right. \\
&\left.\left. \phantom{\frac{1}{2} W( ig_{\sigma,E'_\sigma} ) \Phi_\sigma ,} - g A_{3,\tau}( x',0) + g \Re( h_{3,\sigma}(x') , g_{\sigma,E'_\sigma} ) \right ] \Phi^{\mathrm{ren}}_{\sigma,E'_\sigma,\tau} \right )_{\mathcal{H}_\tau},
\end{split}
\end{equation}
whereas
\begin{equation}\label{eq:E'tau}
\begin{split}
E'_\tau = &\frac{1}{m} \left ( \Phi^{\mathrm{ren}}_{\tau,E'_\sigma} , \left [ P_3 - d\Gamma (k_3) - \Phi ( k_3 g_{\tau,E'_\sigma} ) - \frac{1}{2} ( k_3 g_{\tau,E'_\sigma} , g_{\tau,E'_\sigma} ) \right.\right. \\
&\left.\left. \phantom{\frac{1}{2} W( ig_{\sigma,E'_\sigma} ) \Phi_\sigma ,} - g A_{3,\tau}( x',0) + g \Re( h_{3,\tau}(x') , g_{\tau,E'_\sigma} ) \right ] \Phi^{\mathrm{ren}}_{\tau,E'_\sigma} \right )_{\mathcal{H}_\tau}.
\end{split}
\end{equation}
The expression into brackets being uniformly bounded with respect to $H^{\mathrm{ren}}_{\sigma,E'_\sigma,\tau}$, one can prove that
\begin{equation}
\begin{split}
& \left \| \left [ P_3 - d\Gamma (k_3) - \Phi ( k_3 g_{\tau,E'_\sigma} ) - \frac{1}{2} ( k_3 g_{\sigma,E'_\sigma} , g_{\sigma,E'_\sigma} ) \right.\right. \\
&\left.\left. \phantom{\frac{1}{2}} - g A_{3,\tau}( x',0) + \Re( h_{3,\sigma}(x') , g_{\sigma,E'_\sigma} ) \right ] \Phi^{\mathrm{ren}}_{\sigma,E'_\sigma,\tau} \right \| \le \mathrm{C},
\end{split}
\end{equation}
and likewise with $\Phi^{\mathrm{ren}}_{\tau,E'_\sigma}$ replacing $\Phi^{\mathrm{ren}}_{\sigma,E'_\sigma,\tau}$. In addition, we have
\begin{equation}
\left | ( k_3 g_{\sigma,E'_\sigma} , g_{\sigma,E'_\sigma} ) - ( k_3 g_{\tau,E'_\sigma} , g_{\tau,E'_\sigma} ) \right | \le \mathrm{C} g^2 \sigma,
\end{equation}
and, similarly,
\begin{equation}
\left \| \left [ \Re( h_{3,\tau}(x') , g_{\tau,E'_\sigma} ) - \Re( h_{3,\sigma}(x') , g_{\sigma,E'_\sigma} ) \right ]  \Phi^{\mathrm{ren}}_{\tau,E'_\sigma} \right \| \le \mathrm{C} |g| \sigma.
\end{equation}
Estimating the difference of \eqref{eq:E'sigma} and \eqref{eq:E'tau} then leads to
\begin{equation}
\left | E'_\sigma - E'_\tau \right | \le \mathrm{C} \left [ \left \| \Phi^{\rm ren}_{\sigma,E'_\sigma,\tau} - \Phi^{\rm ren}_{\tau,E'_\sigma} \right \|_{\mathcal{H}_\tau} + g^2 \sigma \right ]
\end{equation}
The statement of the lemma now follows by choosing the non-degenerate ground states $\Phi^{\rm ren}_{\sigma,E'_\sigma,\tau}$ and $\Phi^{\rm ren}_{\tau,E'_\sigma}$ in such a way that
\begin{equation}
\left \| \Phi^{\rm ren}_{\sigma,E'_\sigma,\tau} - \Phi^{\rm ren}_{\tau,E'_\sigma} \right \|_{\mathcal{H}_\tau} \le \mathrm{C} \left \| P^{\rm ren}_{\sigma,E'_\sigma,\tau} - P^{\rm ren}_{\tau,E'_\sigma} \right \|.
\end{equation}
Note that this choice is indeed possible due to the non-degeneracy of the ground states $\Phi^{\rm ren}_{\sigma,E'_\sigma,\tau}$ and $\Phi^{\rm ren}_{\tau,E'_\sigma}$.
\end{proof}
For $g,P_3,\sigma,\rho$ as above, let us define the operator $\nabla H^{\rm ren}_{\tau,\rho}$  by
\begin{equation*}
\begin{split}
\nabla H^{\rm ren}_{\sigma,\rho} &= \frac{1}{m} W( ig_{\sigma,\rho} ) \big [ P_3 - d \Gamma(k_3) - gA_{3,\sigma}(x',0) \big ] W( ig_{\sigma,\rho} ) ^* \\
& = \frac{1}{m} \bigg [ÊP_3 - d\Gamma(k_3) - \Phi( k_3 g_{\sigma,\rho} ) - \frac{1}{2} ( k_3 g_{\sigma,\rho} , g_{\sigma,\rho} ) \\
& \phantom{ \nabla H^{\rm ren}_{\sigma,E'_\sigma,\tau} = \frac{1}{m} \bigg [ P_3 - d\Gamma } \quad - g A_{3,\sigma}( x',0 ) + g \Re ( h_{3,\sigma}(x') , g_{\sigma,\rho} ) \bigg ].
\end{split}
\end{equation*}
\begin{lemma}\label{lm:2nd_step}
Let $\Gamma_{\sigma,\mu}$ be the curve $\Gamma_{\sigma,\mu} = \{ \mu \sigma e^{i \nu}, \nu \in [0,2\pi[ \}$. There exist $g_0 > 0$, $\sigma_0>0$, $\mu>0$ and $P_0>0$, such that for all $0 < |g| \le g_0$, $|P_3| \le P_0$, for all $\sigma > 0$ and $\tau>0$ such that $\sigma/2 \le \tau \le \sigma \le \sigma_0$,
\begin{equation}\label{eq:2nd_step}
\begin{split}
 \left \| P^{\rm ren}_{\sigma,E'_\sigma,\tau} - P^{\rm ren}_{\tau,E'_\sigma} \right \| \le \mathrm{C} & |g|^{1/2} \sigma^{1/2} \sup_{z \in \Gamma_{\sigma,\mu}} \bigg [ 1 + \Big | \Big ( \left ( \nabla H^{\rm ren}_{\sigma,E'_\sigma} - E'_\sigma \right ) \Phi^{\rm ren}_{\sigma,E'_\sigma} , \\
& \big [ H^{\rm ren}_{\sigma,E'_\sigma} - E_\sigma - z \big ]^{-1} \left ( \nabla H^{\rm ren}_{\sigma,E'_\sigma} - E'_\sigma \right ) \Phi^{\rm ren}_{\sigma,E'_\sigma} \Big ) \Big |^{\frac{1}{2}} \bigg ],
\end{split}
\end{equation}
where $\mathrm{C}$ is a positive constant.
\end{lemma}
\begin{proof}
By \cite[Lemma II.11]{BFP},
\begin{equation}\label{eq:lemma_BFP}
\left \| P^{\rm ren}_{\sigma,E'_\sigma,\tau} - P^{\rm ren}_{\tau,E'_\sigma} \right \| = \left | \left ( \Phi^{\rm ren}_{\sigma,E'_\sigma,\tau} , [ P^{\rm ren}_{\sigma,E'_\sigma,\tau} - P^{\rm ren}_{\tau,E'_\sigma} ] \Phi^{\rm ren}_{\sigma,E'_\sigma,\tau} \right ) \right |^{1/2}.
\end{equation}
It follows from Lemma \ref{lm:gap1} and Proposition \ref{prop:gap} that $\mathrm{Gap}( H^{\rm ren}_{\sigma,E'_\sigma,\tau} ) \ge \sigma / 8$ and $\mathrm{Gap}( H^{\rm ren}_{\tau,E'_\sigma} ) \ge \tau / 8 \ge \sigma / 16$. Therefore, since $| E_\sigma - E_\tau | \le \mathrm{C} |g| \sigma$ by Lemma \ref{lm:|Etau-Esigma|}, we can write
\begin{equation*}\label{eq:proj_resolvents}
P^{\rm ren}_{\sigma,E'_\sigma,\tau} - P^{\rm ren}_{\tau,E'_\sigma} = \frac{i}{2\pi} \oint_{\Gamma_{\sigma,\mu}} \left ( \left [ H^{\rm ren}_{\sigma,E'_\sigma,\tau} - E_\sigma - z \right ]^{-1} - \left [ H^{\rm ren}_{\tau,E'_\sigma} - E_\sigma - z \right ]^{-1} \right ) dz,
\end{equation*}
provided $\mu<1/16$ and $|g|$ is sufficiently small. Expanding $\left [ H^{\rm ren}_{\tau,E'_\sigma} - E_\sigma - z \right ]^{-1}$ into a (convergent) Neumann series yields
\begin{equation*}\label{eq:resolvents}
\begin{split}
 P^{\rm ren}_{\sigma,E'_\sigma,\tau} - P^{\rm ren}_{\tau,E'_\sigma} = \frac{i}{2\pi} \sum_{n\ge 1} \oint_{\Gamma_{\sigma,\mu}} & (-1)^n \left [ H^{\rm ren}_{\sigma,E'_\sigma,\tau} - E_\sigma - z \right ]^{-1} \\
& \left ( \left [ H^{\rm ren}_{\tau,E'_\sigma} - H^{\rm ren}_{\sigma,E'_\sigma,\tau} \right ] \left [ H^{\rm ren}_{\sigma,E'_\sigma,\tau} - E_\sigma - z \right ]^{-1} \right )^n dz.
\end{split}
\end{equation*}
Let us compute the difference $H^{\rm ren}_{\tau,E'_\sigma} - H^{\rm ren}_{\sigma,E'_\sigma,\tau}$ explicitly. We have:
\begin{equation*}
\begin{split}
&H^{\rm ren}_{\sigma,E'_\sigma,\tau} = \frac{1}{2m} \sum_{j=1,2} \Big ( p_j - e a_j( x' ) - g A_{j,\sigma}(x',0) + g \Re ( h_{j,\sigma}(x') , g_{\sigma,E'_\sigma} )  \Big )^2 \\
& + \frac{m}{2} ( \nabla H^{\mathrm{ren}}_{\sigma,E'_\sigma} )^2 - \frac{e}{2m} \sigma_3b(x') - \frac{g}{2m} \sigma \cdot \Big ( B_\sigma(x',0) - \Re ( \tilde h_{\sigma}(x') , g_{\sigma,E'_\sigma} ) \Big ) \\
& + V(x') + H_f + \Phi ( |k| g_{\sigma,E'_\sigma} ) + \frac{1}{2} ( |k| g_{\sigma,E'_\sigma} , g_{\sigma,E'_\sigma} ) - \frac{g^2}{2m}( \Lambda^2 - \sigma^2 ),
\end{split}
\end{equation*}
and
\begin{equation*}
\begin{split}
&H^{\rm ren}_{\tau,E'_\sigma} = \frac{1}{2m} \sum_{j=1,2} \Big ( p_j - e a_j( x' ) - g A_{j,\tau}(x',0) + g \Re ( h_{j,\tau}(x') , g_{\tau,E'_\sigma} )  \Big )^2 \\
& + \frac{m}{2} ( \nabla H^{\mathrm{ren}}_{\tau,E'_\sigma} )^2 - \frac{e}{2m} \sigma_3b(x') - \frac{g}{2m} \sigma \cdot \Big ( B_\sigma(x',0) - \Re ( \tilde h_{\tau}(x') , g_{\tau,E'_\sigma} ) \Big ) \\
& + V(x') + H_f + \Phi ( |k| g_{\tau,E'_\sigma} ) + \frac{1}{2} ( |k| g_{\tau,E'_\sigma} , g_{\tau,E'_\sigma} ) - \frac{g^2}{2m}( \Lambda^2 - \tau^2 ).
\end{split}
\end{equation*}
Let us decompose:
\begin{equation}\label{eq:diff_Ham}
\begin{split}
& H^{\rm ren}_{\tau,E'_\sigma} - H^{\rm ren}_{\sigma,E'_\sigma,\tau} = [a] + [b] + [c] + [d] + [e], 
\end{split}
\end{equation}
with
\begin{equation*}
\begin{split}
[a] = \frac{1}{m} \sum_{j=1,2} & \Big ( - g A_{j,\tau}^\sigma(0,0) + g \Re ( h_{j,\tau}(0) , g_{\tau,E'_\sigma}^\sigma ) \Big ) \\
& \times \Big ( p_j - e a_j( x' ) - g A_{j,\sigma}(x',0) + g \Re ( h_{j,\sigma}(x') , g_{\sigma,E'_\sigma} )  \Big ), \phantom{aaaaaaaaaaaaaaaaaa}
\end{split}
\end{equation*}
\begin{equation*}
\begin{split}
&[b] = \frac{1}{2m} \sum_{j=1,2} \Big ( - g A_{j,\tau}^\sigma(x',0) + g \Re ( h_{j,\tau}(x') , g_{\tau,E'_\sigma}^\sigma ) \Big )^2 - \frac{g^2}{2m}( \sigma^2 - \tau^2 )\\
&\quad +\frac{1}{2m} \Big ( - \Phi( k_3 g_{\tau,E'_\sigma}^\sigma ) - \frac{1}{2} ( k_3 g_{\tau,E'_\sigma}^\sigma , g_{\tau,E'_\sigma}^\sigma ) - g A_{3,\tau}^\sigma(x',0) + g \Re ( h_{3,\tau}(x') , g_{\tau,E'_\sigma}^\sigma ) \Big )^2, \phantom{aaaaaaaaaaaaaaaaaa} \\
&\quad + \frac{g}{2m} \sigma \cdot \Re \bigg ( \tilde h_{\tau}(x') - \tilde h_\sigma(x') , g_{\tau,E'_\sigma} \bigg )
\end{split}
\end{equation*}
\begin{equation*}
\begin{split}
[c] = & \frac{1}{m} \sum_{j=1,2} \Big ( - g ( A_{j,\tau}^\sigma(x',0) - A_{j,\tau}^\sigma(0) ) + g \Re ( h_{j,\tau}(x') - h_{j,\tau}(0) , g_{\tau,E'_\sigma}^\sigma ) \Big ) \\
&\phantom{\frac{1}{m} \sum_{j=1,2}} \times \Big ( p_j - e a_j( x' ) - g A_{j,\sigma}(x',0) + g \Re ( h_{j,\sigma}(x') , g_{\sigma,E'_\sigma} )  \Big ) \phantom{aaaaaaaaaaaaaaaaaa} \\
& - g E'_\sigma [ A_{3,\tau}^\sigma(x',0) - A_{3,\tau}^\sigma( 0 , 0 ) ] + g E'_\sigma  \Re ( h_{3,\tau}(x') - h_{3,\tau}(0) , g_{\tau,E'_\sigma}^\sigma ), \phantom{aaaaaaaaaaaaaaaaaaaaaaaaaaaaaaaaaaaaaaaaaa}
\end{split}
\end{equation*}
\begin{equation*}
[d] = g E'_\sigma ( h_{3,\tau}(0) , g_{\tau,E'_\sigma}^\sigma ) - \frac{1}{2} E'_\sigma ( k_3 g_{\tau,E'_\sigma}^\sigma , g_{\tau,E'_\sigma}^\sigma ),\phantom{aaaaaaaaaaaaaaaaaaaaaaaaaaaaaaaaaaaaaaaaaaaaaaaaaaaaaaa}
\end{equation*}
\begin{equation*}
\begin{split}
& [e]=\frac{1}{2} \Big ( - \Phi( k_3 g_{\tau,E'_\sigma}^\sigma ) - \frac{1}{2} ( k_3 g_{\tau,E'_\sigma}^\sigma , g_{\tau,E'_\sigma}^\sigma ) - g A_{3,\tau}^\sigma(x',0) + g \Re ( h_{3,\tau}(x') , g_{\tau,E'_\sigma}^\sigma ) \Big ) \\
& \phantom{[5]\frac{1}{2m}}  \times \Big ( \nabla H^{\rm ren}_{\sigma,E'_\sigma} - E'_\sigma \Big )  + \frac{1}{2} \Big ( \nabla H^{\rm ren}_{\sigma,E'_\sigma} - E'_\sigma \Big ) \\
& \phantom{ [5]\frac{1}{2m} } \times \Big ( - \Phi( k_3 g_{\tau,E'_\sigma}^\sigma ) - \frac{1}{2} ( k_3 g_{\tau,E'_\sigma}^\sigma , g_{\tau,E'_\sigma}^\sigma ) - g A_{3,\tau}^\sigma(x',0) + g \Re ( h_{3,\tau}(x') , g_{\tau,E'_\sigma}^\sigma ) \Big ). \phantom{aaaaa}
\end{split}
\end{equation*}
Note that we have added and subtracted $E'_\sigma$, using the identity $( E'_\sigma k_3 - |k| ) g_{\sigma,E'_\sigma} = - g E'_\sigma h_{3,\sigma} (0)$ and likewise with $g_{\tau,E'_\sigma}$ replacing $g_{\sigma,E'_\sigma}$. Let us now consider, for some $n \ge 1$,
\begin{equation}\label{eq:estimate_Neumann}
\begin{split}
\oint_{\Gamma_{\sigma,\mu}} \bigg (  \Phi^{\rm ren}_{\sigma,E'_\sigma,\tau} , &\left [ H^{\rm ren}_{\sigma,E'_\sigma,\tau} - E_\sigma - z \right ]^{-1} \\
&\left ( \left [ H^{\rm ren}_{\tau,E'_\sigma} - H^{\rm ren}_{\sigma,E'_\sigma,\tau} \right ] \left [ H^{\rm ren}_{\sigma,E'_\sigma,\tau} - E_\sigma - z \right ]^{-1} \right )^n \Phi^{\rm ren}_{\sigma,E'_\sigma,\tau} \bigg ).
\end{split}
\end{equation}
We insert \eqref{eq:diff_Ham} into the right-hand side of \eqref{eq:estimate_Neumann}, thus obtaining a sum of terms that we estimate separately. We claim that all the terms where at least one of the operators $[a]$, $[b]$, or $[c]$ appear, are bounded by $\mathrm{C} \sigma (\mathrm{C}' |g| )^n $ where $\mathrm{C},\mathrm{C}'$ are two positive constants. The latter can be proven by means of rather standard estimates involving pull-through formulas (see for instance \cite{BFS,Pizzo,BFP,CFP}), so we shall not give all the details. Let us still emphasize that in order to deal with $[a]$ or $[c]$ we need to use  the exponential decay of $\Phi^{\rm ren}_{\sigma,E'_\sigma,\tau}$ in $x'$ (proven in \cite[Appendix A]{AGG2}). This is the main difficulty we encounter compared to the proof of \cite{CFP}. In order to overcome it, we adapt a method due to \cite{Sigal} (see also \cite[Section 5]{AFFS}). Let us give an example: Consider
\begin{equation}\label{eq:example}
\bigg (  \Phi^{\rm ren}_{\sigma,E'_\sigma,\tau} , [e] \left [ H^{\rm ren}_{\sigma,E'_\sigma,\tau} - E_\sigma - z \right ]^{-1} [a] \left [ H^{\rm ren}_{\sigma,E'_\sigma,\tau} - E_\sigma - z \right ]^{-1} [e] \Phi^{\rm ren}_{\sigma,E'_\sigma,\tau} \bigg ).
\end{equation}
We shall take advantage of the identity
\begin{equation}\label{eq:[H,x]}
\Big ( p_j - e a_j( x' ) - g A_{j,\sigma}(x',0) + g \Re ( h_{j,\sigma}(x') , g_{\sigma,E'_\sigma} )  \Big ) =  2i \left [ H^{\rm ren}_{\sigma,E'_\sigma,\tau} , x'_j \right ]
\end{equation}
which holds in the sense of quadratic forms on $D( H^{\mathrm{ren}}_{\sigma,E'_\sigma,\tau} ) \cap D( x'_j )$. 
The field operator $A_{j,\sigma}^\tau(0,0) = \Phi( h_{j,\sigma}^\tau )$ in $[a]$ decompose into a sum of a creation operator and an annihilation operator that are estimated separately. Take for instance the creation operator. Using a pull-through formula, we have to bound:
\begin{equation}\label{eq:example2}
\begin{split}
g \int h_{j,\tau}^\sigma( \bold k ) \bigg (  \Phi^{\rm ren}_{\sigma,E'_\sigma,\tau} , & [e] a^*( \bold k ) \left [ H^{\rm ren}_{\sigma,E'_\sigma,\tau}(P_3-k_3) - E_\sigma + |k| - z \right ]^{-1} \\
& \left [ H^{\rm ren}_{\sigma,E'_\sigma,\tau} , x'_j \right ] \left [ H^{\rm ren}_{\sigma,E'_\sigma,\tau} - E_\sigma - z \right ]^{-1} [e] \Phi^{\rm ren}_{\sigma,E'_\sigma,\tau} \bigg ) d \bold k.
\end{split}
\end{equation}
Let $\gamma>0$ be such that $\| e^{ \gamma \langle x' \rangle} \Phi^{\rm ren}_{\sigma,E'_\sigma,\tau} \| < \infty$. Undoing the commutator $[ H^{\rm ren}_{\sigma,E'_\sigma,\tau} , x'_j ]$ gives two terms. We write the first one under the form
\begin{equation*}
\begin{split}
 g \int h_{j,\tau}^\sigma( \bold k ) \bigg ( & \left ( H^{\rm ren}_{\sigma,E'_\sigma,\tau} - E_\sigma \right )  \left [ H^{\rm ren}_{\sigma,E'_\sigma,\tau}(P_3-k_3) - E_\sigma + |k| - \bar z \right ]^{-1} a( \bold k ) [e]^* \Phi^{\rm ren}_{\sigma,E'_\sigma,\tau} , \\
&  x'_j e^{- \gamma \langle x' \rangle } e^{\gamma \langle x' \rangle } \left [ H^{\rm ren}_{\sigma,E'_\sigma,\tau} - E_\sigma - z \right ]^{-1} e^{-\gamma \langle x' \rangle } [e] e^{\gamma \langle x' \rangle} \Phi^{\rm ren}_{\sigma,E'_\sigma,\tau} \bigg ) d \bold k.
\end{split}
\end{equation*}
Now we have the following estimates:
\begin{align}
& \left \| e^{\gamma \langle x' \rangle } \left [ H^{\rm ren}_{\sigma,E'_\sigma,\tau} - E_\sigma - z \right ]^{-1} e^{-\gamma \langle x' \rangle } [e] e^{\gamma \langle x' \rangle} \Phi^{\rm ren}_{\sigma,E'_\sigma,\tau} \right \| \le \mathrm{C} |g|, \\
& \left \| x'_j e^{ - \gamma \langle x' \rangle } \right \| \le \mathrm{C}, \\
& \left \| \left [ H^{\rm ren}_{\sigma,E'_\sigma,\tau}(P_3-k_3) - E_\sigma + |k| - z \right ]^{-1} \left ( H^{\rm ren}_{\sigma,E'_\sigma,\tau} - E_\sigma \right ) \right \| \le \mathrm{C}, \label{eq:estimate_example3} \\
& \left \| a( \bold k ) [e]^* \Phi^{\rm ren}_{\sigma,E'_\sigma,\tau} \right \| \le \mathrm{C} |g| |k|^{-1/2}. \label{eq:estimate_example4}
\end{align}
Note that in \eqref{eq:estimate_example3} and \eqref{eq:estimate_example4}, we used that $\tau \le |k| \le \sigma$, and thus in particular that $a( \bold k ) \Phi^{\mathrm{ren}}_{\sigma,E'_\sigma,\tau} = 0$. Since the other term coming from the commutator $[ H^{\rm ren}_{\sigma,E'_\sigma,\tau} , x'_j ]$ can be estimated in the same way, this yields
\begin{equation}
| \eqref{eq:example2} | \le \mathrm{C} |g|^3 \int | h_{j,\tau}^\sigma( \bold k ) | |k|^{-1/2} d \bold k \le \mathrm{C} |g|^3 \sigma^2.
\end{equation}
Taking into account the factor $\sigma$ coming from the integration in \eqref{eq:estimate_Neumann} would finally lead to our claim in the case of the example \eqref{eq:example}. 
The same holds for the terms containing $[c]$ at least once (except that the use of \eqref{eq:[H,x]} is then not required). Besides, since $[d]$ is constant,
\begin{equation*}\label{eq:estimate_n}
\oint_{\Gamma_{\sigma,\mu}} \left ( \Phi^{\rm ren}_{\sigma,E'_\sigma,\tau} , \left [ H^{\rm ren}_{\sigma,E'_\sigma,\tau} - E_\sigma - z \right ]^{-1} \left ( \left [ d \right ] \left [ H^{\rm ren}_{\sigma,E'_\sigma,\tau} - E_\sigma - z \right ]^{-1} \right )^n \Phi^{\rm ren}_{\sigma,E'_\sigma,\tau} \right )=0.
\end{equation*}
Therefore it remains to consider the terms containing only $[d]$ or $[e]$, with $[e]$ appearing at least in one factor. One can prove that this leads to
\begin{equation*}
\begin{split}
& \left \| P^{\rm ren}_{\sigma,E'_\sigma,\tau} - P^{\rm ren}_{\tau,E'_\sigma} \right \| \le \mathrm{C} |g|^{1/2}\sigma^{1/2} \sup_{z \in \Gamma_{\sigma,\mu}} \Bigg [ 1 + \sigma^{-1} \bigg \| \left | H^{\rm ren}_{\sigma,E'_\sigma,\tau} - E_\sigma - z \right |^{-1/2} \\ 
& \qquad \bigg ( - \Phi( k_3 g_{\tau,E'_\sigma}^\sigma ) - \frac{1}{2} ( k_3 g_{\tau,E'_\sigma}^\sigma , g_{\tau,E'_\sigma}^\sigma ) - g A_{3,\tau}^\sigma(x',0) + g \Re ( h_{3,\tau}(x') , g_{\tau,E'_\sigma}^\sigma )  \bigg ) \\
& \qquad \left [ \nabla H^{\rm ren}_{\sigma,E'_\sigma} - E'_\sigma \right ] \Phi^{\rm ren}_{\sigma,E'_\sigma,\tau} \bigg \|  \Bigg ].
\end{split}
\end{equation*}
Using again the exponential decay of $\Phi^{\mathrm{ren}}_{\sigma,E'_\sigma}$ in $x'$, we may replace $\Re ( f_{3,\tau}(x') , g_{\tau,E'_\sigma}^\sigma )$ with $\Re ( f_{3,\tau}(0) , g_{\tau,E'_\sigma}^\sigma )$ in the previous expression. Proceeding then as in \cite[Lemma A.3]{CFP}, since both
$$
( \Phi( k_3 g_{\tau,E'_\sigma}^\sigma ) + g A_{3,\tau}^\sigma(x',0) ) ( \nabla H^{\rm ren}_{\sigma,E'_\sigma} - E'_\sigma ) \Phi^{\rm ren}_{\sigma,E'_\sigma,\tau} \quad \text{and} \quad ( \nabla H^{\rm ren}_{\sigma,E'_\sigma} - E'_\sigma ) \Phi^{\rm ren}_{\sigma,E'_\sigma,\tau}
$$
are orthogonal to $\Phi^{\mathrm{ren}}_{\sigma,E'_\sigma,\tau}$, we obtain Inequality \eqref{eq:2nd_step} (notice in particular that $\sigma_0$ and $\mu$ must be fixed sufficiently small to pass from the last estimate to \eqref{eq:2nd_step}).
\end{proof}
\textsc{Proof of Proposition \ref{prop:E'tau-E'sigma}} 
To conclude the proof of Proposition \ref{prop:E'tau-E'sigma}, in view of Lemmata \ref{lm:E'_GS} and \ref{lm:2nd_step}, it suffices to show that
\begin{equation*}
\Big | \Big ( \left ( \nabla H^{\rm ren}_{\sigma,E'_\sigma} - E'_\sigma \right ) \Phi^{\rm ren}_{\sigma,E'_\sigma} ,  \big [ H^{\rm ren}_{\sigma,E'_\sigma} - E_\sigma - z \big ]^{-1} \left ( \nabla H^{\rm ren}_{\sigma,E'_\sigma} - E'_\sigma \right ) \Phi^{\rm ren}_{\sigma,E'_\sigma} \Big ) \Big | \le \frac{ \mathrm{C}_\delta }{ |g| \sigma^{2\delta} },
\end{equation*}
for any $z \in \Gamma_{\sigma,\mu}$ and any $\delta>0$. This corresponds to the bound (IV.68) in \cite{CFP} and can be proven in the same way as in \cite[Subsection IV.5, step (4)]{CFP}, using an induction procedure. We therefore refer the reader to \cite{CFP} for a proof.
\hfill $\square$ \\\\ 
\textsc{Proof of Theorem \ref{thm:regularity}}
Fix $P_3$ and $k_3$ such that $|P_3| \le P_0$, $|P_3+k_3| \le P_0$. One can see that there exist positive constants $\mathrm{C}_0$ and $\mathrm{C}$ such that, for any $0<\beta<1$ and $\sigma \ge \mathrm{C}_0 |k_3|^\beta$,
\begin{equation}
\left | E'_\sigma( P_3+k_3 ) - E'_\sigma(P_3) \right | \le \mathrm{C}  |k_3|^{\frac{1}{2}(1-\beta)}.
\end{equation}
This can be proven by estimating $| E'_\sigma( P_3+k_3 ) - E'_\sigma(P_3) |$ in terms of $ \| \Phi_\sigma( P_3 + k_3) - \Phi_\sigma(P_3) \|$, then using the second resolvent equation to estimate $\| [ H_\sigma(P_3+k_3) - z ]^{-1} - [H_\sigma(P_3) -z]^{-1} \|$. Now, for $\sigma \le \mathrm{C}_0 |k_3|^\beta$, we use Proposition \ref{prop:E'tau-E'sigma}, which yields
\begin{equation*}
\begin{split}
&\left | E'_\sigma( P_3+k_3 ) - E'_\sigma(P_3) \right | \\
&\le \left | E'_\sigma( P_3+k_3 ) - E'_{\mathrm{C}_0 |k_3|^\beta}(P_3 + k_3) \right | + \left | E'_{\mathrm{C}_0 |k_3|^\beta}( P_3+k_3 ) - E'_{\mathrm{C}_0 |k_3|^\beta}(P_3) \right | \\
&\quad+ \left | E'_{\sigma}( P_3 ) - E'_{\mathrm{C}_0|k_3|^\beta}(P_3) \right | \\
&\le \mathrm{C}_\delta \left [ |k_3|^{\frac{1}{2}(1-\beta)} + |k_3|^{\frac{1}{2}\beta(1-\delta)} \right ].
\end{split}
\end{equation*}
The theorem follows by choosing $\beta = [2-\delta]^{-1}$.
\hfill $\square$

\bibliographystyle{amsalpha}

\end{document}